\documentclass[lettersize,journal]{IEEEtran}
\usepackage{amsmath,amsfonts}
\usepackage{array}
\usepackage[caption=false,font=normalsize,labelfont=sf,textfont=sf]{subfig}
\usepackage{textcomp}
\usepackage{stfloats}
\usepackage{url}
\usepackage{verbatim}
\usepackage{graphicx}
\usepackage{cite}
\usepackage[hidelinks]{hyperref}	% 超链接
\usepackage[ruled]{algorithm2e} 	% 算法伪代码包
\usepackage{xcolor}					% 改变公式中字符的颜色

\begin{document}
\newtheorem{assumption}{Assumption}
\newtheorem{definition}{Definition}
\newtheorem{proof}{Proof}
\newtheorem{remark}{Remark}
\newtheorem{proposition}{Proposition}
\newtheorem{lemma}{Lemma}
\newtheorem{theorem}{Theorem}
\newtheorem{corollary}{Corollary}

\title{Data-driven Moving Horizon Estimation for Angular Velocity of Space Noncooperative Target in Electromagnetic De-tumbling Mission}
\author{Xiyao Liu,~Haitao Chang,~Fei Hui,~Zhenyu Lu,~Yizhai Zhang,~Panfeng Huang, \IEEEmembership{Senior Member,~IEEE}
	\thanks{This work has been submitted to the IEEE for possible publication. Copyright may be transferred without notice, after which this version may no longer be accessible.}
	\thanks{This research was supported by the National Natural Science Foundation of China (Grant No.62403072).}
	\thanks{X. Liu and F. Hui is with School of Electronics and Control Engineering, Chang’an University, Xi’an, Shaanxi 710064, China, E-mail: (liuxiyao@chd.edu.cn;feihui@chd.edu.cn). H. Chang, ,Y. Zhang and P. Huang are with the Research Center for Intelligent Robotics, School of Astronautics, Northwestern Polytechnical University, Xi’an 710072, China, E-mail: (htchang@nwpu.edu.cn;zhangyizhai@nwpu.edu.cn; pfhuang@nwpu.edu.cn; ); Z. Lu is with School of Automation Science and Engineering, South China University of Technology, Guangzhou, Guangdong Province, E-mail: (luzhenyu@scut.edu.cn).}
}

\markboth{Journal of \LaTeX\ Class Files, April~2026}%
%\markboth{This work has been submitted to the IEEE for possible publication. Copyright may be transferred without notice, after which this version may no longer be accessible.}
{Shell \MakeLowercase{\textit{et al.}}: A Sample Article Using IEEEtran.cls for IEEE Journals}


\maketitle

\begin{abstract}
Angular velocity estimation is critical for electromagnetic eddy current de-tumbling of noncooperative space targets. 
However, unknown model of the noncooperative target and few observation data make the model-based estimation methods challenged.
In this paper, a Data-driven Moving Horizon Estimation (Dd-MHE) method is proposed to estimate the angular velocity of the noncooperative target by electromagnetic torque. 
In this method, model-free state estimation of the angular velocity can be achieved using only one historical trajectory data that satisfies the rank condition. 
With local linear approximation, the Willems’ fundamental lemma is extended to nonlinear autonomous systems, and the rank condition for the historical trajectory data is deduced.
Then, a data-driven moving horizon estimation algorithm based on the $M$-step Lyapunov function is designed, and the time-discount robust stability of the algorithm is given.
In order to illustrate the effectiveness of the proposed algorithm, experiments and simulations are performed to estimate the angular velocity in electromagnetic eddy current de-tumbling with only electromagnetic torque measurement.
\end{abstract}

\def\abstractname{Note to Practitioners}
\begin{abstract}
	The motivation of this paper is to estimate the angular velocity of high-speed, non-cooperative targets in electromagnetic de-tumbling mission at relatively short distances, without the need for magnetic sensors or electromagnets on the targets. To achieve this, we propose a data-driven moving horizon estimation strategy that leverages Willems' fundamental lemma and moving horizon estimation methods based on both historical measurements and current electromagnetic torque measurements. Through experiments and simulations, we have confirmed the effectiveness and accuracy of our methods, which can be integrated with existing non-cooperative target angular velocity measurement systems and extended to information estimation problems beyond unknown models. Going forward, our methods could be combined with optical and inertial measurement to achieve comprehensive information estimation in robotics and aerospace applications.
\end{abstract}


\section{Introduction}
\label{sec:introduction}
\IEEEPARstart{A}{ngular} velocity information is critical for trajectory planning and control of robots and spacecraft \cite{11010080}. Despite their widespread use, cameras suffer from a limited field of view, making angular velocity measurement of unmarked, high-speed rotating targets a challenging task \cite{10629195}. Magnetic sensors offer a viable non-contact solution for attitude estimation \cite{10599514}. However, their application is limited by the requirement for on-board electromagnets or magnetic sensors mounted on the target. This prerequisite is often unmet in non-cooperative target electromagnetic de-tumbling missions \cite{liu2022robust}.

In the de-tumbling missions \cite{liu2022robust}, the electromagnetic torque acting on the chaser can be directly measured, and an inherent physical relationship exists between this torque and the target’s angular velocity. The target’s angular velocity can thus be estimated via the measured electromagnetic torque. Unlike conventional magnetic sensing or actuation technologies \cite{10752578}, this approach requires no magnetic sensors or permanent magnets installed on the space non-cooperative target; the only available data for state estimation is the eddy current-induced electromagnetic torque on the chaser.

However, non-cooperative targets have unknown physical properties and complex geometries, making it impossible to accurately model the mapping between the target’s angular velocity and the induced electromagnetic torque. This fundamentally hinders model-based state estimation methods, such as the Unscented Kalman Filter \cite{11408264}, and motivates the exploration of model-free state estimation approaches.

In recent years, data-driven methods have advanced rapidly, with three mainstream branches: learning-based, Koopman operator-based, and Willems’ fundamental lemma-based methods. Learning-based data-driven state estimation relies on various deep learning networks, hyperparameter tuning, and corresponding training strategies \cite{jin2021new}. Such methods have been widely applied in power systems, but require abundant data to learn an accurate system model \cite{10643712}. The Koopman operator provides a framework to represent nonlinear systems in a linear form, complementing classical and statistical perspectives on dynamic systems \cite{otto2021koopman}. Surana \textit{et al.} \cite{surana2020koopman} first proposed a Koopman operator-based linear observer for nonlinear systems, and extended this method to constrained state estimation. Guo \textit{et al.} \cite{guo2021koopman} combined the computational efficiency of Koopman state estimation with the generality of kernel embedding via Random Fourier Features, and verified the algorithm on a mobile robot platform.

The above methods all require large datasets to train accurate deep learning networks or resolve full system eigenvalues. However, in electromagnetic eddy current de-tumbling missions \cite{liu2021eddy}, experiments are unrepeatable, and only limited historical data from a single trajectory is available. A data-driven method rooted in behavioral theory—Willems’ fundamental lemma \cite{willems2005note}—offers a viable solution for angular velocity estimation in this scenario. This lemma states that for a controllable linear time-invariant system, any trajectory of the system can be reconstructed from corresponding input-output data, as long as the historical input satisfies the persistently exciting condition \cite{faulwasser2023behavioral}. Building on this, van Waarde \textit{et al.} \cite{9062331} extended the lemma to state-space systems with multiple datasets, and Schmitz \textit{et al.} \cite{9739697} further generalized it to linear descriptor systems. Both extensions effectively reduced the data volume requirement for data-driven analysis, and verified the lemma’s superiority in small-data scenarios.

For lemma-based state estimation, Mishra \textit{et al.} \cite{10591180} developed data-driven state reconstruction and a Kalman filter-like algorithm for linear systems in both deterministic and stochastic settings. Wolff \textit{et al.} \cite{10453955} proposed robust data-driven moving horizon estimation for linear discrete-time systems with noisy offline and online data, and proved its practical robust exponential stability. Wei \textit{et al.} \cite{WEI2026112614} designed a distributed data-driven unknown-input observer for continuous-time linear systems subject to unknown inputs and disturbances. However, these methods are either limited to linear systems or cannot handle system disturbances, while the angular velocity estimation problem in eddy current de-tumbling is inherently nonlinear and disturbed. Therefore, a new data-driven estimation method tailored to this mission is urgently needed, which can adapt to unknown nonlinear systems with disturbances under small-data constraints.

Motivated by the above challenges, this paper proposes a data-driven moving horizon estimation algorithm for angular velocity estimation of non-cooperative targets with completely unknown dynamics. Based on local linear approximation, the algorithm can address the state estimation problem of unknown nonlinear systems in a local region. In addition, the method supports angular velocity estimation under constraint conditions, and its time-discounted robust stability under local linear approximation is rigorously proved \cite{knufer2020time}.

The main contributions of this paper are as follows:
\begin{enumerate}
	\item A novel model-free data-driven method for non-cooperative target angular velocity estimation is proposed, which requires only a small volume of data. Unlike learning-based \cite{chen2021data} and Koopman-based \cite{guo2021koopman} methods, this approach is built on Willems’ fundamental lemma from behavioral theory, requiring only one historical trajectory satisfying the rank condition to complete data-driven estimation.
	\item A dedicated data-driven moving horizon estimation (Dd-MHE) algorithm is developed, which is applicable to nonlinear autonomous systems via local linear approximation. Compared with existing methods \cite{adachi2021dual,turan2021data,wolff2021data}, this algorithm can handle state estimation of nonlinear systems, while guaranteeing time-discounted robust stability under constraint satisfaction and disturbances. In addition, a pole placement method for unknown systems is provided, which gives the lower bound of parameters to ensure the stability of the Dd-MHE algorithm.
\end{enumerate}

The remainder of this paper is organized as follows. Section \ref{sec2} establishes the eddy current de-tumbling dynamic model, and presents the extension of Willems’ fundamental lemma to autonomous systems. Section \ref{sec3} details the proposed data-driven moving horizon estimation algorithm and its time-discounted robust stability analysis. Section \ref{sec:4} verifies the effectiveness of the proposed method via simulations and experiments of eddy current de-tumbling angular velocity estimation.

\textit{Notation}: $\boldsymbol{I}_n$ represents $n \times n$-dimensional identity matrix, $\otimes$ denotes the Kronecker product. $\boldsymbol{1}_{n}$ represents an ${n}$-dimensional column vector with all entries equal to 1. $\boldsymbol{0}$ denotes a matrix with corresponding dimensions and all entries equal to 0. The set of integers is expressed as $\mathbb{I}$, and the set of integers greater than or equal to $a$ or between $a$ and $b$ are expressed as $\mathbb{I}_{\geq a}$ and $\mathbb{I}_{[a,b]}$ respectively. The maximal eigenvalues of matrix $\boldsymbol{P}$ are denoted by $\lambda_{max}(\boldsymbol{P})$.
A data sequence $\boldsymbol{u}_{[i, j]}$ is a column vector with $\boldsymbol{u}_{[i, j]} = [\boldsymbol{u}(i)^\top \ \  \boldsymbol{u}(i+1)^\top \ \cdots \ \boldsymbol{u}(j)^\top]^\top$. For any matrix $\boldsymbol{A}$, $\boldsymbol{A}^{\dagger}$ denotes the pseudoinverse of $\boldsymbol{A}$.
For any vector $\boldsymbol{a} = [a_1 \ a_2 \ a_3]^\top$, $(\cdot)^{\times}$ denotes the skew-symmetric matrix as follows:
\begin{equation}
	\boldsymbol{a}^{\times} = 
	\begin{bmatrix}
		0 			& -a_{3} 	&  a_{2} \\
		a_{3} 	& 0 			&  -a_{1} \\
		-a_{2} & a_{1} 	&  0
	\end{bmatrix}.
\end{equation}

The Hankel matrix is given by:
\begin{definition}
	\label{Def:1}
	The Hankel matrix $\boldsymbol{H}_k(\boldsymbol{u}_{[i, j]})$ with depth $k \in \mathbb{I}_{\geq 1}$ of any vector sequence $\boldsymbol{u}_{[i, j]}$ is defined as
	$$
	\boldsymbol{H}_k(\boldsymbol{u}_{[i, j]}) \! =\!\!\! \begin{bmatrix}
		\boldsymbol{u}(i) & \boldsymbol{u}(i+1) & \cdots & \boldsymbol{u}(j - k + 1) \\
		\boldsymbol{u}(i+1) & \boldsymbol{u}(i+2) & \cdots & \boldsymbol{u}(j-k+2) \\
		\vdots & \vdots & \ddots & \vdots \\
		\boldsymbol{u}(i+k-1) & \boldsymbol{u}(i+k) & \cdots & \boldsymbol{u}(j) \notag
	\end{bmatrix}\!.$$
\end{definition}

\section{Problem setup and preliminaries}
\label{sec2}
In this section, we present and linearize the target's attitude dynamics model in the eddy current de-tumbling according to \cite{liu2021eddy} and then extend Willems' fundamental lemma to linear autonomous systems with offsets in conjunction with \cite{berberich2022linear}.

\subsection{Attitude dynamics model in eddy current de-tumbling}
According to \cite{liu2021eddy}, the attitude dynamics model of the unknown target in eddy current de-tumbling is
\begin{equation}
	\label{Euler rotational equations of target}
	\boldsymbol{J}_t\dot{\boldsymbol{\omega}}_t
	+{\boldsymbol{\omega}}_t^{\times} \boldsymbol{J}_t{\boldsymbol{\omega}}_t
	=\boldsymbol{\tau}_{t}(\boldsymbol{\omega}_t),
\end{equation}
where $\boldsymbol{J}_t$ denotes the inertia tensor of the target, ${\boldsymbol{\omega}}_t$ is the target's angular velocity, ${\boldsymbol{\tau}_{t}}(\boldsymbol{\omega}_t)$ denotes the eddy current de-tumbling torque.

Since the target is unknown, the parameter $\boldsymbol{J}_t$ and function ${\boldsymbol{\tau}_{t}}(\boldsymbol{\omega}_t)$ in the above equation are not available. Therefore, the data-driven method is needed to estimate the target's angular velocity $\boldsymbol{\omega}_t$. The measured output is the eddy current de-tumbling torque on the chaser given as
\begin{align}
	\label{tau_c}
	\boldsymbol{y} = \boldsymbol{\tau}_{c}(\boldsymbol{\omega}_t).
\end{align}
Similarly, $\boldsymbol{\tau}_{c}(\boldsymbol{\omega}_t)$ is a function of ${\boldsymbol{\omega}}_t$ and is also not available. Note that one can obtain approximate equations for $\boldsymbol{\tau}_{t}(\boldsymbol{\omega}_t)$ and $\boldsymbol{\tau}_{c}(\boldsymbol{\omega}_t)$ by (6) and (9) in \cite{gomez2017guidance}. However, since the target is completely unknown, the parameters in these equations are also not available. Therefore, for the sake of brevity, we directly consider $\boldsymbol{\tau}_{t}(\boldsymbol{\omega}_t)$ and $\boldsymbol{\tau}_{c}(\boldsymbol{\omega}_t)$ as unknown functions.

Combining \eqref{Euler rotational equations of target}-\eqref{tau_c} and the forward-Euler method, one can directly deduce the following standard nonlinear discrete-time state space model:
\begin{align}
	\label{state space model}
	{\boldsymbol{x}}(t+1) = \boldsymbol{f}(\boldsymbol{x}(t)), \boldsymbol{y}(t) = \boldsymbol{h}(\boldsymbol{x}(t)),
\end{align}
where $\boldsymbol{x} = {\boldsymbol{\omega}}_t$ denotes the state, the discrete-time state function $\boldsymbol{f}(\boldsymbol{x}(t)) =\boldsymbol{x}(t) + T_s \boldsymbol{J}_t^{-1}\left( \boldsymbol{\tau}_{t}(\boldsymbol{x}) - \boldsymbol{x}_t^{\times}\boldsymbol{J}_t\boldsymbol{x}_t \right) $, $T_s$ is the sample interval, and the output function $\boldsymbol{h}(\boldsymbol{x}(t)) = \boldsymbol{\tau}_{c}(\boldsymbol{x}(t))$. 

In the following content, we need the linearization model of \eqref{state space model}. Thus, we define the linear autonomous system with offsets resulting from the linearization of \eqref{state space model} at point $\tilde{\boldsymbol{x}}$ as
\begin{align}
	\label{Eq:LDAS}
	\begin{aligned}
		&\boldsymbol{x}(t+1) = \boldsymbol{A}_{\tilde{\boldsymbol{x}}} \boldsymbol{x}(t) + \boldsymbol{e}_{\tilde{\boldsymbol{x}}} + \boldsymbol{w}(t),  \\
		&\boldsymbol{y}(t) = \boldsymbol{C}_{\tilde{\boldsymbol{x}}} \boldsymbol{x}(t) + \boldsymbol{r}_{\tilde{\boldsymbol{x}}} + \boldsymbol{v}(t),
	\end{aligned}
\end{align}
where $\boldsymbol{A}_{\tilde{\boldsymbol{x}}}=\left.\frac{\partial \boldsymbol{f}}{\partial \boldsymbol{x}}\right|_{\tilde{\boldsymbol{x}}}$, $\boldsymbol{e}_{\tilde{\boldsymbol{x}}} = \boldsymbol{f}(\tilde{\boldsymbol{x}}) - \boldsymbol{A}_{\tilde{\boldsymbol{x}}} \tilde{\boldsymbol{x}}$, $\boldsymbol{C}_{\tilde{\boldsymbol{x}}} =\left.\frac{\partial \boldsymbol{h}}{\partial \boldsymbol{x}}\right|_{\tilde{\boldsymbol{x}}}$, $\boldsymbol{r}_{\tilde{\boldsymbol{x}}} =\boldsymbol{h}(\tilde{\boldsymbol{x}})-\boldsymbol{C}_{\tilde{\boldsymbol{x}}} \tilde{\boldsymbol{x}}$, and $\boldsymbol{w}(t) = \boldsymbol{f}(\boldsymbol{x}(t)) - \boldsymbol{A}_{\tilde{\boldsymbol{x}}} \boldsymbol{x}(t) - \boldsymbol{e}_{\tilde{\boldsymbol{x}}}$ is the state function linearization error, 
$\boldsymbol{v}(t) = \boldsymbol{h}(\boldsymbol{x}(t)) - \boldsymbol{C}_{\tilde{\boldsymbol{x}}} \boldsymbol{x}(t) {{-}} \boldsymbol{r}_{\tilde{\boldsymbol{x}}}$ is the output function linearization error. 
Note that $\boldsymbol{A}_{\tilde{\boldsymbol{x}}}$, $\boldsymbol{e}_{\tilde{\boldsymbol{x}}}$, $\boldsymbol{C}_{\tilde{\boldsymbol{x}}}$ and $\boldsymbol{r}_{\tilde{\boldsymbol{x}}}$ in \eqref{Eq:LDAS} are unknown, and $\boldsymbol{w}(t)$ and $\boldsymbol{v}(t)$ can be regarded as process disturbance and measurement noise respectively.

\subsection{Extended Willems' fundamental lemma}
To derive the data-driven method used in this paper, we extend the Willems' fundamental lemma \cite{berberich2022linear} to the linear autonomous system with offsets.

\begin{lemma}\label{Prop:1}
	For a linear autonomous system with offsets: 
	\begin{align}
		\label{Eq:LAS}
		\begin{aligned}
			&\boldsymbol{x}(t+1) = \boldsymbol{A} \boldsymbol{x}(t) + {{\boldsymbol{e}_o}}, \\
			&\boldsymbol{y}(t) = \boldsymbol{C} \boldsymbol{x}(t) + {{\boldsymbol{r}_o}},
		\end{aligned}
	\end{align}
	{where $\boldsymbol{e}_o$ and $\boldsymbol{r}_o$ denote the offsets of the state equation and the output equation, respectively. }
	Suppose {one data sequence} $ \{ \boldsymbol{x}_{[0, T]}, \boldsymbol{y}_{[0, T-1]} \}$ is a trajectory of \eqref{Eq:LAS}, and satisfies the rank condition as
	\begin{align}
		\label{rank:H1}
		{\mathrm{rank}} \left( \begin{bmatrix}
			\boldsymbol{H}_1(\boldsymbol{x}_{[0,T-L]}) \\
			\boldsymbol{1}_{T-L+1}^{\top}
		\end{bmatrix}
		\right) = n + 1,
	\end{align}
	{where $n$ is the dimension of the system state $\boldsymbol{x}$, and the data length $T$ satisfies that $T \geq L + n$.}
	Then, any data sequence marked as $\{{\boldsymbol{x}}'_{[0, L]},{\boldsymbol{y}}'_{[0, L-1]}\}$ is a trajectory of the system \eqref{Eq:LAS},
	%	{where the operator ${(\cdot)}'$ is used to distinguish the two data sequences,} 
	if and only if there exists $\boldsymbol{\alpha} \in \mathbb{R}^{T-L+1}$ such that
	\begin{align}\label{proposition:1}
		\begin{bmatrix}
			\boldsymbol{H}_{L+1}(\boldsymbol{x}_{[0, T]}) \\
			\boldsymbol{H}_{L}(\boldsymbol{y}_{[0, T-1]}) \\
			\boldsymbol{1}^{\top}_{T-L+1}
		\end{bmatrix} \boldsymbol{\alpha} = 
		\begin{bmatrix}
			{{{\boldsymbol{x}}'_{[0, L]}}} \\
			{{\boldsymbol{y}'_{[0, L-1]}}} \\
			1
		\end{bmatrix}.
	\end{align}
	
\end{lemma}
%The proof of this lemma can be obtained in \cite{liu2023data}.
\begin{proof}	
	\textbf{Proof of ``if'':}
	
	Before starting the proof, we define the following matrix:
	\begin{align}
		\begin{aligned}
			&\boldsymbol{\varPhi}_K = \begin{bmatrix}
				\boldsymbol{I}_n^{\top} &
				\boldsymbol{A}^{\top} &
				\cdots^{\top} &
				(\boldsymbol{A}^{K-1})^{\top}
			\end{bmatrix}^{\top}, \\
			&\boldsymbol{\varGamma}_K = \begin{bmatrix}
				\boldsymbol{C}^{\top} &
				(\boldsymbol{CA})^{\top} &
				\cdots &
				(\boldsymbol{C}\boldsymbol{A}^{K-1})^{\top}
			\end{bmatrix}^{\top},\\ 
			&\boldsymbol{\varPsi}_K = \begin{bmatrix}
				\boldsymbol{0}^{\top} &
				\boldsymbol{I}_n^{\top} &
				\cdots &
				(\sum_{i=0}^{K-2}\boldsymbol{A}^{i})^{\top}
			\end{bmatrix}^{\top},\\
			&\boldsymbol{\varPi}_K = \begin{bmatrix}
				\boldsymbol{0}^{\top} &
				\boldsymbol{C}^{\top} &
				\cdots &
				(\boldsymbol{C}\sum_{i=0}^{K-2}\boldsymbol{A}^{i})^{\top}
			\end{bmatrix}^{\top}.
		\end{aligned}
	\end{align}
	
	Since $\{ \boldsymbol{x}_{[0, T]}, \boldsymbol{y}_{[0, T-1]} \}$ is a trajectory of the system \eqref{Eq:LAS}, the following equation can be derived as
	\begin{align}\label{App:12}
		\begin{bmatrix}
			\boldsymbol{H}_{L+1}(\boldsymbol{x}_{[0, T]}) \\
			\boldsymbol{H}_{L}(\boldsymbol{y}_{[0, T-1]})
		\end{bmatrix} \! \! = \! \!
		\begin{bmatrix}
			\boldsymbol{\varPhi}_{L+1} \\
			\boldsymbol{\varGamma}_{L}
		\end{bmatrix}  \! \! \boldsymbol{x}_L \! \! + \! \!
		\begin{bmatrix}
			\boldsymbol{\varPsi}_{L+1} \\
			\boldsymbol{\varPi}_L
		\end{bmatrix} \!  \boldsymbol{e}_L \! \! + \! \! 
		\begin{bmatrix}
			\boldsymbol{0} \\
			\boldsymbol{I}_{p,L}
		\end{bmatrix} \! \boldsymbol{r}_L,
	\end{align}
	where $\boldsymbol{x}_L = \boldsymbol{H}_1(\boldsymbol{x}_{[0,T-L]})$, $\boldsymbol{e}_L = \boldsymbol{1}_{T-L+1}^{\top} \otimes {{\boldsymbol{e}_o}}$, $\boldsymbol{r}_L = \boldsymbol{1}_{T-L+1}^{\top} \otimes {{\boldsymbol{r}_o}}$, $\boldsymbol{I}_{p,L} = \boldsymbol{I}_p\otimes \boldsymbol{1}_{L}$.
	
	By substituting \eqref{App:12} into \eqref{proposition:1}, one can obtain
	\begin{align}\label{App:13}
		\begin{bmatrix}
			{{{\boldsymbol{x}}'_{[0, L]}}} \\
			{{\boldsymbol{y}'_{[0, L-1]}}}
		\end{bmatrix} 
		& \! \! =  \! \! 
		\begin{bmatrix}
			\boldsymbol{\varPhi}_{L+1} \\
			\boldsymbol{\varGamma}_{L}
		\end{bmatrix}\boldsymbol{x}_L \boldsymbol{\alpha}  + 	
		\begin{bmatrix}
			\boldsymbol{\varPsi}_{L+1} \\
			\boldsymbol{\varPi}_L
		\end{bmatrix}\boldsymbol{e}_L \boldsymbol{\alpha} + 
		\begin{bmatrix}
			\boldsymbol{0} \\
			\boldsymbol{I}_{p,L}
		\end{bmatrix} \boldsymbol{r}_L  \boldsymbol{\alpha} \notag \\
		& \! \! = \! \!
		\begin{bmatrix}
			\boldsymbol{\varPhi}_{L+1} \\
			\boldsymbol{\varGamma}_{L}
		\end{bmatrix}\boldsymbol{x}_L \boldsymbol{\alpha}  + 	
		\begin{bmatrix}
			\boldsymbol{\varPsi}_{L+1} \\
			\boldsymbol{\varPi}_L
		\end{bmatrix}{{\boldsymbol{e}_o}} + 
		\begin{bmatrix}
			\boldsymbol{0} \\
			\boldsymbol{I}_{p,L}
		\end{bmatrix}{{\boldsymbol{r}_o}}.
	\end{align}
	In the above derivation, $\boldsymbol{1}^{\top}_{T-L+1} \boldsymbol{\alpha} = 1$ is adopted to make $\boldsymbol{e}_L \boldsymbol{\alpha} = {{\boldsymbol{e}_o}}$ and $\boldsymbol{r}_L \boldsymbol{\alpha} = {\boldsymbol{e}_o}$.
	
	Equation \eqref{rank:H1} directly derives that $\boldsymbol{x}_L = \boldsymbol{H}_1(\boldsymbol{x}_{[0,T-L]})$ has rank $n$. Thus, \eqref{App:13} means that $\{ {{{\boldsymbol{x}}'_{[0, L]}}}, {{\boldsymbol{y}'_{[0, L-1]}}} \}$ is a trajectory of the system \eqref{Eq:LAS} with initial condition ${{\boldsymbol{x}'(0)}} := \boldsymbol{H}_1(\boldsymbol{x}_{[0,T-L]}) \boldsymbol{\alpha}$.
	
	\textbf{Proof of ``Only if'':}
	
	Since $\boldsymbol{H}_1(\boldsymbol{x}_{[0,T-L]})$ has rank $n$, it implies the existence of a vector $\boldsymbol{\alpha} \in \mathbb{R}^{T-L+1}$ such that
	\begin{align}
		\begin{bmatrix}
			\boldsymbol{H}_1(\boldsymbol{x}_{[0,T-L]}) \\
			\boldsymbol{1}^{\top}_{T-L+1}
		\end{bmatrix}\boldsymbol{\alpha} = 
		\begin{bmatrix}
			{{\boldsymbol{x}'(0)}} \\
			1
		\end{bmatrix}.
	\end{align}
	
	Considering that $\{ {{{\boldsymbol{x}}'_{[0, L]}}}, {{\boldsymbol{y}'_{[0, L-1]}}} \}$ is a trajectory of system \eqref{Eq:LAS}, the following derivation is available:
	\begin{align}\label{App:only if}
		\begin{bmatrix}
			{{{\boldsymbol{x}}'_{[0, L]}}} \\
			{{\boldsymbol{y}'_{[0, L-1]}}}
		\end{bmatrix} 
		&=
		\begin{bmatrix}
			\boldsymbol{\varPhi}_{L+1} \\
			\boldsymbol{\varGamma}_{L}
		\end{bmatrix}{{\boldsymbol{x}'(0)}}  + 	
		\begin{bmatrix}
			\boldsymbol{\varPsi}_{L+1} \\
			\boldsymbol{\varPi}_L
		\end{bmatrix}{{\boldsymbol{e}_o}} + 
		\begin{bmatrix}
			\boldsymbol{0} \\
			\boldsymbol{I}_{p,L}
		\end{bmatrix}{{\boldsymbol{r}_o}} \notag \\
		& \overset{\eqref{App:13}}{=}
		\begin{bmatrix}
			\boldsymbol{\varPhi}_{L+1} \\
			\boldsymbol{\varGamma}_{L}
		\end{bmatrix}\boldsymbol{x}_L \boldsymbol{\alpha} \! + 	\!
		\begin{bmatrix}
			\boldsymbol{\varPsi}_{L+1} \\
			\boldsymbol{\varPi}_L
		\end{bmatrix}\boldsymbol{e}_L \boldsymbol{\alpha} \! + \!
		\begin{bmatrix}
			\boldsymbol{0} \\
			\boldsymbol{I}_{p,L}
		\end{bmatrix} \boldsymbol{r}_L\boldsymbol{\alpha} \notag \\
		& \overset{\eqref{App:12}}{=}
		\begin{bmatrix}
			\boldsymbol{H}_{L+1}(\boldsymbol{x}_{[0, T]}) \\
			\boldsymbol{H}_{L}(\boldsymbol{y}_{[0, T-1]})
		\end{bmatrix}\boldsymbol{\alpha}.
	\end{align}
	
	Thereby, the proof of \textit{Lemma \ref{Prop:1}} is completed.  $\Box$
\end{proof}
%The proof of \textit{Proposition \ref{Prop:1}} is provided in Appendix A.

\section{Data-driven MHE and Robust stability}
\label{sec3}
In this section, we construct the robust data-driven {Moving Horizon Estimation (MHE)} optimization problem for linear autonomous systems with offsets based on the historical state/output trajectory and provide the \textbf{Algorithm \ref{algorithm1}}. Subsequently, we give the robust stability of \textbf{Algorithm \ref{algorithm1}} with the $M$-step Lyapunov function\cite{schiller2023lyapunov}.

\subsection{Data-driven MHE Algorithm}
\label{sec3_1}
For the linearized dynamic model \eqref{Eq:LDAS}, the traditional MHE problem \cite{rao2001constrained,schiller2023lyapunov} with horizon length $M_t = \min\{t,M\}$ and $M \in \mathbb{I}_{\geq 1}$ is
\begin{subequations}
	\label{opti1}
	\begin{align}
		%	\begin{aligned}
			&\min_{\hat{\boldsymbol{x}}(t-M_t|t),\hat{\boldsymbol{w}}(\cdot|t),\hat{\boldsymbol{v}}(\cdot|t)}  J_{L}(\hat{\boldsymbol{x}}(t-M_t|t),\hat{\boldsymbol{w}}(\cdot|t),\hat{\boldsymbol{v}}(\cdot|t),t)  \\
			&\text{s.t.}  \ 
			\hat{\boldsymbol{x}}(j \! + \! 1|t) \! = \! \boldsymbol{A}_{\tilde{\boldsymbol{x}}} \hat{\boldsymbol{x}}(j|t) + 
			\hat{\boldsymbol{w}}(j|t), \  j  \in   \mathbb{I}_{[t-M_t,t-1]}, \label{opti1_a}\\
			&\ \quad \boldsymbol{y}(j) = \boldsymbol{C}_{\tilde{\boldsymbol{x}}} \hat{\boldsymbol{x}}(j|t) + \hat{\boldsymbol{v}}(j|t), \  j  \in   \mathbb{I}_{[t-M_t,t-1]} \label{opti1_b}, \\
			&\ \quad \hat{\boldsymbol{w}}(j|t)  \in  \mathbb{W}, \ \hat{\boldsymbol{v}}(j|t)  \in  \mathbb{V}, \  j  \in  \mathbb{I}_{[t-M_t,t-1]}, \label{opti1_c} \\
			&\ \quad \hat{\boldsymbol{x}}(j|t) \in  \mathbb{X}, \ j  \in  \mathbb{I}_{[t-M_t,t]}, \label{opti1_d}
			%	\end{aligned}
	\end{align}
\end{subequations}
where the cost function of the traditional MHE optimization problem is
\begin{align}
	\label{Eq:OBJ}
	&J_{L}(\hat{\boldsymbol{x}}(t-M_t|t),\hat{\boldsymbol{w}}(\cdot|t),\hat{\boldsymbol{v}}(\cdot|t),t) \notag \\
	& \qquad = \|\hat{\boldsymbol{x}}(t - M_t|t)  -  \bar{\boldsymbol{x}}(t - M_t) \|_{\boldsymbol{P}}^2  \notag \\
	& \qquad \quad  +  \sum_{j=1}^{M_t} \left( \|\hat{\boldsymbol{w}}(t -  j|t) \|_{\boldsymbol{Q}}^2 + \|\hat{\boldsymbol{v}}(t  -  j|t) \|_{\boldsymbol{R}}^2 \right), 
\end{align}
and	$\bar{\boldsymbol{x}}(t - M_t)$ is the prior information for the state ${\boldsymbol{x}}(t - M_t)$. $\hat{\boldsymbol{x}}(\cdot|t) = \{ \hat{\boldsymbol{x}}(j|t) \}_{j=t-M_t}^{j=t}$, $\hat{\boldsymbol{w}}(\cdot|t) = \{\hat{\boldsymbol{w}}(j|t)\}_{j=t-M_t}^{j=t-1}$ and $\hat{\boldsymbol{v}}(\cdot|t) = \{\hat{\boldsymbol{v}}(j|t)\}_{j=t-M_t}^{j=t-1}$ {denote} a sequence of $M_t$ optimized state, disturbance and noise estimates, respectively. $\boldsymbol{P}$, $\boldsymbol{Q}$ and $\boldsymbol{R}$ are the weight matrices of prior, disturbance, and noise, respectively. 
$\mathbb{X}$, $\mathbb{W}$ and $\mathbb{V}$ are state and noise constraint sets, respectively, if given.

By leveraging the equivalence between Willems' fundamental lemma and the linearized dynamic model \eqref{Eq:LDAS}, the data-driven optimization form \eqref{opti1} can be obtained by replacing the dynamic model with Willems' fundamental lemma and $\boldsymbol{P} = 2\rho^{M_t}$, $\boldsymbol{Q} = \boldsymbol{0}$, $\boldsymbol{R} = \rho^{j-1} \mu$.

As stated in \textit{Lemma \ref{Prop:1}} and Ref.\cite{wolff2021data}, noise-free system state/output historical data is required to realize data-driven estimation. In eddy current de-tumbling, the historical state data can be obtained by visual-based systems such as LIght Detection And Ranging (LIDAR) sensors, monocular and stereo cameras, etc \cite{cassinis2019review}. In this paper, the historical data is defined as $\{\boldsymbol{x}_{[-T, 0]},\boldsymbol{y}_{[-T, -1]}\}$. 

At time $t \in \mathbb{I}_{\geq 1}$, the proposed Data-driven MHE (Dd-MHE) optimization problem is
\begin{subequations}
	\label{opti1}
	\begin{align}
		%	\begin{aligned}
			&\min_{\hat{\boldsymbol{x}}(\cdot|t),\boldsymbol{\alpha}(t),\hat{\boldsymbol{v}}(\cdot|t)} \  J_{L}(\hat{\boldsymbol{x}}(t-M_t|t),\hat{\boldsymbol{v}}(\cdot|t),t)  \\
			&\text{s.t.} \!
			\begin{bmatrix}
				\boldsymbol{H}_{M_t}\left(\boldsymbol{y}_{[-T, -1]}\right) \\
				\boldsymbol{H}_{M_t+1}\left({\boldsymbol{x}}_{[-T, 0]}\right) \\
				\boldsymbol{1}^{\top}_{T-M_t+1}
			\end{bmatrix} \! \! \boldsymbol{\alpha}(t) \! = \! \!
			\begin{bmatrix}
				\boldsymbol{y}_{[t-M_t, t-1]} - \hat{\boldsymbol{v}}_{[-M_t,-1]}(t) \\
				\hat{\boldsymbol{x}}_{[-M_t,0]}(t)  \\
				1
			\end{bmatrix}\! \!, \label{opti1:1}  \\
			& \quad \ \hat{\boldsymbol{x}}(j|t)  \in \mathbb{X},\ j  \in  \mathbb{I}_{[t-M_t,t]}, \label{opti1:2}\\  
			& \quad \ \hat{\boldsymbol{v}}(j|t)  \in \mathbb{V},\ j  \in   \mathbb{I}_{[t-M_t,t-1]}, \label{opti1_c}
			%	\end{aligned}
	\end{align}
\end{subequations}
where the cost function of this optimization problem is
\begin{align}
	\label{Eq:OBJ}
	J_{L}(\hat{\boldsymbol{x}}(t \! - \! M_t|t),\hat{\boldsymbol{v}}(\cdot|t),t)=& 2\rho^{M_t}\|\hat{\boldsymbol{x}}(t \! - \! M_t|t) \! - \! \bar{\boldsymbol{x}}(t \! - \! M_t) \|^2 \notag \\
	+ & \sum_{j=1}^{M_t}\rho^{j-1} \mu \|\hat{\boldsymbol{v}}(t-j|t) \|^2, 
\end{align}
and $\boldsymbol{y}_{[t-M_t, t-1]}$ {denotes} the measurement output sequence at time interval $[t-M_t, t-1]$. $\hat{\boldsymbol{x}}_{[-M_t,0]}(t) = [\hat{\boldsymbol{x}}(t-M_t|t)^\top,\dots,\hat{\boldsymbol{x}}(t|t)^\top ]^\top$ and $\hat{\boldsymbol{v}}_{[-M_t,-1]}(t) = [\hat{\boldsymbol{v}}(t-M_t|t)^\top,\dots,\hat{\boldsymbol{v}}(t-1|t)^\top ]^\top$ {denote} a stack of the optimized state and noise, respectively.
The parameters $\rho$ and $\mu$ are set to ensure robust stability, which will be described in detail in the subsection \ref{Robust Stability}{, and the horizon length $M$ is limited to satisfying the inequality \eqref{def:kappa}.}
Therefore, we give the following Dd-MHE algorithm for the linearized dynamic model \eqref{Eq:LDAS}, as shown in \textbf{Algorithm \ref{algorithm1}}.

{Note that in \textbf{Algorithm 1}, when $t \leq M$, $\bar{\boldsymbol{x}}(t-M_t) = \bar{\boldsymbol{x}}(0)$ where $\bar{\boldsymbol{x}}(0)$ is the prior information of state $\boldsymbol{x}$ at time $t=0$ and can be taken as a point in the state constraint set $\mathbb{X}$. When $t > M$, $\bar{\boldsymbol{x}}(t-M_t) = \hat{\boldsymbol{x}}(t-M_t)$, where $\hat{\boldsymbol{x}}(t-M_t)$ is the estimated value of ${\boldsymbol{x}}(t-M_t)$. This choice of updating prior information $\bar{\boldsymbol{x}}(t-M_t)$ is typically called filtering update \cite{schiller2023lyapunov,rawlings2017model}.
\begin{algorithm}[]  
	\label{algorithm1}
	\caption{Dd-MHE for the system \eqref{Eq:LDAS}}
	%	\LinesNumbered 
	%	\textbf{Initialization:}
	%	\justifying
	\textbf{Data Collection:} Collect historical data $\{\boldsymbol{x}_{[-T, 0]}$, $\boldsymbol{y}_{[-T, -1]}\}$ through sensors and ensure that the data satisfy the rank condition \eqref{rank:H1}.\\
	%	\justifying
	\textbf{Data Preprocessing:} 
	Execute \textbf{Algorithm \ref{algorithm2}} and obtain $\rho_0$ and $\mu_0$; 
	Given parameters $\rho \in (\rho_0, 1)$, $\mu > \mu_0$ and $M$ that satisfies \eqref{def:kappa}; Given the constraint sets $\mathbb{X}$ and $\mathbb{V}$, and the priori estimate state $\bar{\boldsymbol{x}}(0)$ of \eqref{opti1}; Given the final time $T_{f}$.\\%输入参数
	\textbf{State Estimation:}\\
	\For{$t \leq T_{f}$}{
		\If{$t \leq M$}{
			Create the optimization problem \eqref{opti1} with $M_t = t$\;
			Let $\bar{\boldsymbol{x}}(t-M_t) = \bar{\boldsymbol{x}}(0)$ and given the measurement output $\boldsymbol{y}_{[t-M_t, t-1]}$ to \eqref{opti1}\;
			Solve \eqref{opti1} and obtain $\hat{\boldsymbol{x}}^*(t|t)$;
		}
		\If{$M < t \leq T_f$}{
			Let $\bar{\boldsymbol{x}}(t-M_t) = \hat{\boldsymbol{x}}(t-M_t)$ and given the measurement output $\boldsymbol{y}_{[t-M_t, t-1]}$ to \eqref{opti1}\;
			Solve \eqref{opti1} and obtain $\hat{\boldsymbol{x}}^*(t|t)$\;
		}
		$\hat{\boldsymbol{x}}(t) = \hat{\boldsymbol{x}}^*(t|t)$; }
\end{algorithm}
{
	\begin{remark}\label{remark 1}
		It is worth noting that the optimization problem \eqref{opti1} does not include the disturbance $\boldsymbol{w}(t)$. The reason is that it is difficult to construct the $M$-step Lyapunov function in \textit{Theorem \ref{algorithm1}} when the disturbance $\boldsymbol{w}(t)$ is included, so that the robust stability of Dd-MHE cannot be deduced theoretically. In addition, if appropriate parameters can be determined, including the disturbance $\boldsymbol{w}(t)$ will improve the performance of Dd-MHE and can cancel \textit{Assumption \ref{assum2}} of this paper. It is an interesting issue to include the disturbance $\boldsymbol{w}(t)$ and choose the appropriate parameters in the future.
	\end{remark}
}
{
	\begin{remark}\label{remark 2}
		In eddy current de-tumbling, the chaser stays close to the target for a long time. For example, the distance to the target surface is about $1 \ \rm{(m)}$ \cite{liu2022robust}. However, the measurement of target's angular velocity based on LIDAR or optical encoder has been a difficult problem when the relative distance is short and the target rotates at high speed\cite{long2022monocular,zhang2022monocular}. 		
		Therefore, the historical data is obtained by LIDAR or optical encoder within their effective range. Then, \textbf{Algorithm \ref{algorithm1}} is used to estimate the target's angular velocity when the distance is close and the target rotates at high speed. In addition, historical data that does not rely on ground truth angular velocity of the target are better, but this issue is still under study. In the latest research, Wolff et al. \cite{10453955} solve the problem that the ground truth data contains noise but still relies on the ground truth historical data. We hope to solve the dependence of this algorithm on ground truth historical data in future research.
	\end{remark}
}

\subsection{Robust Stability}
\label{Robust Stability}
In order to illustrate the robust stability of Dd-MHE, we introduce the time-discounted robust stability \cite{knufer2020time} of the following convolution sum form. Let $\boldsymbol{e}(t) = \boldsymbol{x}(t) - \hat{\boldsymbol{x}}(t)$ and $\bar{\boldsymbol{e}}(0) = \boldsymbol{x}(0) - \bar{\boldsymbol{x}}(0)$, then we have the following definition.
\begin{definition}({Def. 1 of \cite{knufer2018robust}}, Def. 2.3 of\cite{allan2021nonlinear}).
	A state estimator is robustly globally exponentially stable (RGES) if there exist $\lambda_x,\lambda_w,\lambda_v \in (0,1)$ and $c_x,c_w,c_v > 0$, such that the resulting state estimate $\hat{\boldsymbol{x}}(t)$ satisfies
	\begin{align}
		\|\boldsymbol{e}(t) \|  \leq  c_x  \lambda_x^t \|  \bar{\boldsymbol{e}}(0) \|  +   \sum_{j = 0}^{t-1}  c_w \lambda_w^{t  -  j  -  1} \|   \boldsymbol{w}(j)  \|
		 \notag \\
		 +   \sum_{j = 0}^{t-1} c_v \lambda_v^{t  -  j  -  1} \|   \boldsymbol{v}(j)  \|,
	\end{align}
	for all $t \in \mathbb{I}_{\geq 0}$, all initial conditions $\boldsymbol{x}(0),\bar{\boldsymbol{x}}(0) \in \mathbb{X}$, and every system trajectory 
	$\{\boldsymbol{x}(t),\boldsymbol{y}(t),\boldsymbol{w}(t),\boldsymbol{v}(t) \}_{t=0}^{t=\infty}$.
\end{definition}

Note that the generalized definition of RGES is the convolution maximum form, that is, $\|\boldsymbol{e}(t) \| \leq \beta_x(\| \bar{\boldsymbol{e}}(0) \|,t) \oplus \max_{j \in \mathbb{I}_{[0,t-1]}}\beta_w(\| \boldsymbol{w}(t) \|,t-j-1) \oplus \max_{j \in \mathbb{I}_{[0,t-1]}} \beta_v(\| \boldsymbol{v}(t) \|,t-j-1)$ with $\beta_x, \beta_w, \beta_v \in \mathcal{KL}$. However, the two forms are equivalent for the exponential stable case, given by Proposition 3.13 of \cite{allan2021robust}. Therefore, for simplicity, we directly use the above convolution sum form.
\begin{remark}
	The definition of robust stability in the time-discounted form, also called as convolution maximization form in \cite{allan2021nonlinear}, is a new tool for dealing with the robust stability of more general estimators, which is discussed in detail in \cite{knufer2020time,allan2021nonlinear,allan2021robust}. Compared with the asymptotic gain form, as in Def. 3 of \cite{allan2019moving}, the robust stability of the time-discounted form can ensure that the robustness will not deteriorate with the increase of the horizon length, etc., as discussed in the abstract of \cite{allan2021robust}. 
	%	Also, compared to the convolution sum form, the convolution maximization form ensures that an arbitrary $\mathcal{KL}$ function can adequately discount bounded sequences in the asymptotic case. Of course, in the exponential case, the convolution sum form is equivalent to the maximization form \cite{allan2021robust}.
\end{remark}

To prove RGES, we give the following assumption as:
\begin{assumption}
	\label{assum:1}
	The matrix pair ($\boldsymbol{A}_{\tilde{\boldsymbol{x}}},\boldsymbol{C}_{\tilde{\boldsymbol{x}}}$) of system \eqref{Eq:LDAS} is observable.
\end{assumption}

Under this assumption, we have the following proposition:
\begin{proposition}
	\label{Prop:2}
	For any two trajectories of system \eqref{Eq:LDAS},  $\{\boldsymbol{x}_1(t),\boldsymbol{y}_1(t),\boldsymbol{w}_1(t),\boldsymbol{v}_1(t),\boldsymbol{e}_1,\boldsymbol{r}_1\}$ and $\{\boldsymbol{x}_2(t),\boldsymbol{y}_2(t),\boldsymbol{w}_2(t),\boldsymbol{v}_2(t),\boldsymbol{e}_2,\boldsymbol{r}_2\}$, there exist $\rho \in (0,1)$ and $\mu > 0$ such that 
	\begin{align}
		\label{norm of delta}
		\| \boldsymbol{x}_{\Delta}(t+1) \|^2 \leq&  \rho \| \boldsymbol{x}_{\Delta}(t) \|^2 + \mu \left(\|\boldsymbol{y}_{\Delta}(t)\|^2 + \| \boldsymbol{r}_{\Delta}\|^2 \right. \notag \\
		& \left. + \| \boldsymbol{v}_{\Delta}(t) \|^2 + \|\boldsymbol{e}_{\Delta}\|^2 + \|\boldsymbol{w}_{\Delta}(t)\|^2 \right),
	\end{align}
	where $\boldsymbol{x}_{\Delta}(t) = \boldsymbol{x}_1(t) - \boldsymbol{x}_2(t)$, $\boldsymbol{y}_{\Delta}(t) = \boldsymbol{y}_1(t) - \boldsymbol{y}_2(t)$, $\boldsymbol{w}_{\Delta}(t) = \boldsymbol{w}_1(t) - \boldsymbol{w}_2(t)$, $\boldsymbol{v}_{\Delta}(t) = \boldsymbol{v}_1(t) - \boldsymbol{v}_2(t)$, $\boldsymbol{e}_{\Delta} = \boldsymbol{e}_1 - \boldsymbol{e}_2$, $\boldsymbol{r}_{\Delta} = \boldsymbol{r}_1 - \boldsymbol{r}_2$.
\end{proposition}
\begin{proof}\label{proof 2}
	According to the superposition property of the system \eqref{Eq:LDAS}, we have
	\begin{align}
		%		\begin{aligned}
			\boldsymbol{x}_{\Delta}(t+1) &= \boldsymbol{A}_{\tilde{\boldsymbol{x}}} \boldsymbol{x}_{\Delta}(t) + \boldsymbol{e}_{\Delta} + \boldsymbol{w}_{\Delta}(t), \label{system:delta1} \\
			\boldsymbol{y}_{\Delta}(t) &= \boldsymbol{C}_{\tilde{\boldsymbol{x}}} \boldsymbol{x}_{\Delta}(t) + \boldsymbol{r}_{\Delta} + \boldsymbol{v}_{\Delta}(t). \label{system:delta2}
			%		\end{aligned}
	\end{align}
	
	%	Since the matrix pair ($\boldsymbol{A}_{\tilde{\boldsymbol{x}}},\boldsymbol{C}_{\tilde{\boldsymbol{x}}}$) is observable, there exists $\boldsymbol{L}$ such that $\boldsymbol{A}_{\boldsymbol{L}} = \boldsymbol{A}_{\tilde{\boldsymbol{x}}} - \boldsymbol{L}\boldsymbol{C}_{\tilde{\boldsymbol{x}}}$ is stable.	
	Thus, we can obtain the following equation by giving a matrix $\boldsymbol{L}$:
\begin{align}
		\begin{aligned}
			\boldsymbol{x}\!_{\Delta}(t\!+\!1) \! = \! & \boldsymbol{A}_{\tilde{\boldsymbol{x}}} \boldsymbol{x}_{\Delta}(t) + \boldsymbol{e}_{\Delta} + \boldsymbol{w}_{\Delta}(t)   \\
			&\! - \!  \boldsymbol{L} \left( \boldsymbol{y}_{\Delta}(t) \! - \! \boldsymbol{C}_{\tilde{\boldsymbol{x}}} \boldsymbol{x}_{\Delta}(t) \! - \! \boldsymbol{r}_{\Delta} \! - \! \boldsymbol{v}_{\Delta}(t) \right)  \\
			\! =\! & \boldsymbol{A}\!_{\boldsymbol{L}} \boldsymbol{x}\!_{\Delta} \! (t)\!  + \! \boldsymbol{e}\!_{\Delta} \! \!  + \! \boldsymbol{w}\!_{\Delta} \! (t) \! \! - \! \! \boldsymbol{L} \! \left( \boldsymbol{y}\!_{\Delta} \! (t)\! \! - \! \boldsymbol{r}\!_{\Delta}  \! \! - \! \boldsymbol{v}\!_{\Delta} \! (t) \right) \!,
		\end{aligned}
	\end{align}
	{where $\boldsymbol{A}_{\boldsymbol{L}} = \boldsymbol{A}_{\tilde{\boldsymbol{x}}} + \boldsymbol{L} \boldsymbol{C}_{\tilde{\boldsymbol{x}}}$.}
	Combining the above results with the Cauchy-Schwarz inequality, $\boldsymbol{x}_{\Delta}(t+1)$ satisfies
	\begin{align}\label{def:rho mu}
		%		\begin{aligned}
			\| \boldsymbol{x}_{\Delta}(t+1) \|^2 \leq &  6\left( \|\boldsymbol{A}_{\boldsymbol{L}} \boldsymbol{x}_{\Delta}(t) \|^2 + \|\boldsymbol{e}_{\Delta}\|^2 + \|\boldsymbol{w}_{\Delta}(t)\|^2 \right. \notag \\ 
			&  \left. + \|\boldsymbol{L}\boldsymbol{y}_{\Delta}(t)\|^2
			+ \| \boldsymbol{L}\boldsymbol{r}_{\Delta}\|^2 + \| \boldsymbol{L}\boldsymbol{v}_{\Delta}(t) \|^2 \right) \notag \\
			\leq&  \rho \| \boldsymbol{x}_{\Delta}(t) \|^2 + \mu \left(\|\boldsymbol{y}_{\Delta}(t)\|^2 + \| \boldsymbol{r}_{\Delta}\|^2 \right. \notag \\
			& \left. + \| \boldsymbol{v}_{\Delta}(t) \|^2 + \|\boldsymbol{e}_{\Delta}\|^2 + \|\boldsymbol{w}_{\Delta}(t)\|^2 \right),
			%		\end{aligned}
	\end{align}
	where $\mu$ and $\rho$ are positive numbers satisfied the constraints:
	\begin{align}\label{constraint rho mu}
		\begin{aligned}
			6\lambda_{max}(\boldsymbol{L}^\top \boldsymbol{L} ) \leq \mu , \quad
			6\lambda_{max}(\boldsymbol{A}_{\boldsymbol{L}}^\top \boldsymbol{A}_{\boldsymbol{L}}) \leq \rho < 1.
		\end{aligned}
	\end{align}
	Thus, \eqref{norm of delta} is derived.  $\Box$
\end{proof}

Note that the pole placement design guarantees that $\rho \in (0,1)$ and can be set as any small value, but correspondingly, the value of $\mu$ will increase. Since the model considered in this paper is unknown, $\mu$ should be taken as large as possible to ensure \eqref{norm of delta}.

{
	\begin{remark}
		In the above proof, we require the parameters $\mu$ and $\rho$ to satisfy \eqref{constraint rho mu}. However, since the system is unknown, the lower bound in \eqref{constraint rho mu} is also unknown. The \textbf{Algorithm \ref{algorithm2}} in Appendix \ref{sec:Appendix A} can calculate the parameters $\rho_0$ and $\mu_0$ that meet the constraints $\rho_0 := 6\lambda_{max}(\boldsymbol{A}_{\boldsymbol{L}}^\top \boldsymbol{A}_{\boldsymbol{L}}) \in (0,1)$ and $\mu_0 := 6\lambda_{max}(\boldsymbol{L}^\top \boldsymbol{L} )$. Thus, \eqref{constraint rho mu} holds for $\mu \geq \mu_0$ and $1 > \rho \geq \rho_0$.
	\end{remark}
}

To ensure the feasibility of the subsequent proof, we require the following assumption:
\begin{assumption}
	\label{assum2}
	For all $t \in \mathbb{I}_{\geq 0}$, $\boldsymbol{x}(t+1) - \boldsymbol{w}(t) \in \mathbb{X}$.
\end{assumption}

This assumption is equivalent to $\boldsymbol{A}_{\tilde{\boldsymbol{x}}} \boldsymbol{x}(t) + \boldsymbol{e}_{\tilde{\boldsymbol{x}}} \in \mathbb{X}$ and makes the state constraint satisfied for all $\boldsymbol{x}(t)$ with $t \in \mathbb{I}_{\geq 0}$.

In what follows, we combine \textit{Lemma \ref{Prop:2}} with \textbf{Theorem 1} of Ref.\cite{schiller2023lyapunov,9805818} to give the conclusion that $W_{\delta}(\hat{\boldsymbol{x}}(t),\boldsymbol{x}(t)) = \|\boldsymbol{x}(t) - \hat{\boldsymbol{x}}(t) \|^2$ is the \textit{M}-step Lyapunov function of Dd-MHE.

\begin{theorem}($M$-step Lyapunov function for data-driven MHE of {Algorithm 1}).
	\label{theorem 1}
	Let Assumption \ref{assum:1} hold. Then, for all $t \in \mathbb{I}_{\geq 0}$, there exist $\rho \in (0,1)$ and $\mu > 0$ such that the state estimate $\hat{\boldsymbol{x}}$ by {Algorithm 1} satisfies 
	\begin{align}\label{theorem:eq 1}
		\|\boldsymbol{x}(t) - \hat{\boldsymbol{x}}(t) \|^2 \leq 4\rho^{M_t}\| \boldsymbol{x}(t-M_t) - \bar{\boldsymbol{x}}(t-M_t) \|^2 \notag \\
		+  \sum_{j=1}^{M_t}\rho^{j-1} \mu( \| \boldsymbol{w}(t-j) \|^2 + 2\| \boldsymbol{v}(t-j)\|^2).
	\end{align}
\end{theorem}
\begin{proof}
	The proof is divided into two parts. In the first part, we consider that the solution of {Algorithm 1} is the trajectory of the system \eqref{Eq:LDAS} and give the following inequality 
	\begin{align}\label{ineq:1}
		%		\begin{aligned}
			\| \boldsymbol{x}(t) - \hat{\boldsymbol{x}}^*(t|t) \|^2 \leq& 2\rho^{M_t}\| \boldsymbol{x}(t-M_t) - \bar{\boldsymbol{x}}(t-M_t) \|^2 \notag  \\
			&+ \sum_{j=1}^{M_t}\rho^{j-1} \mu( \| \boldsymbol{w}(t \! - \! j) \|^2 \! + \! \| \boldsymbol{v}(t \! - \! j)\|^2) \notag \\
			&+ J_{L}(\hat{\boldsymbol{x}}^*(t-M_t|t),\hat{\boldsymbol{v}}^*(\cdot|t),t).
			%		\end{aligned}
	\end{align}
	In the second part, we complete the proof of \textit{Theorem \ref{theorem 1}} by constructing a feasible solution to the optimization problem \eqref{opti1} and inequality \eqref{ineq:1}.
	
	\textbf{Part I:} Let the real trajectory of the system \eqref{Eq:LDAS} be the first trajectory, i.e., $\boldsymbol{x}_1(t) = \boldsymbol{x}(t)$ , $\boldsymbol{y}_1(t)= \boldsymbol{y}(t)$, $\boldsymbol{w}_1(t)=\boldsymbol{w}(t)$, $\boldsymbol{v}_1(t)= \boldsymbol{v}(t)$, $\boldsymbol{e}_1 = \boldsymbol{e}_{\tilde{\boldsymbol{x}}}$, $\boldsymbol{r}_1 = \boldsymbol{r}_{\tilde{\boldsymbol{x}}}$.
	In addition, according to \textit{Lemma \ref{Prop:1}}, $\hat{\boldsymbol{x}}_{[-M_t,0]}^*(t)$ and $\boldsymbol{y}_{[t-M_t, t-1]} - \hat{\boldsymbol{v}}_{[-M_t,-1]}^*(t)$ can be regarded as a trajectory of the system \eqref{Eq:LDAS} without disturbance and noise. Thus, we have the second trajectory as $\boldsymbol{x}_2(t-j) = \hat{\boldsymbol{x}}^*(t-j|t)$ with $j \in \mathbb{I}_{[0,M_t]}$, and $\boldsymbol{y}_2(t-j) = \boldsymbol{y}(t-j) - \hat{\boldsymbol{v}}^*(t-j|t)$, $\boldsymbol{w}_2(t-j)=\boldsymbol{0}$, $\boldsymbol{v}_2(t-j)= \boldsymbol{0}$ with $j \in \mathbb{I}_{[1,M_t]}$, and $\boldsymbol{e}_2 = \boldsymbol{e}_{\tilde{\boldsymbol{x}}}$, $\boldsymbol{r}_2 = \boldsymbol{r}_{\tilde{\boldsymbol{x}}}$.
	
	Substituting the above two trajectories into \textit{Proposition \ref{Prop:2}}, we get the following inequality as
	\begin{subequations}\label{cont:1}
		\begin{align}
			&\|\boldsymbol{e}^*(t|t)  \|^2 \leq \rho \| \boldsymbol{e}^*(t-1|t) \|^2  + \mu \left(\| \hat{\boldsymbol{v}}^*(t-1|t)\|^2 \right. \notag \\
			& \qquad \qquad \qquad \qquad \left. + \| \boldsymbol{v}(t-1)\|^2 +  \| \boldsymbol{w}(t-1) \|^2 \right) \! ,\label{cont:1a} \\
			& {\|\boldsymbol{e}^*(t-1|t)  \|^2 \leq}  
			{\rho \| \boldsymbol{e}^*(t-2|t) \|^2  + \mu \left(\| \hat{\boldsymbol{v}}^*(t-2|t)\|^2 \right.} \notag \\
			& \qquad \qquad \qquad \qquad {\left. + \| \boldsymbol{v}(t-2)\|^2 +  \| \boldsymbol{w}(t-2) \|^2 \right) \! ,} \label{cont:1b}\\
			& \qquad \qquad \qquad \qquad {\vdots} \notag \\
			& {\|\boldsymbol{e}^*(t \! - \! M_t \! + \! 1|t)  \|^2 \leq}  
			{\rho \| \boldsymbol{e}^*(t \! - \! M_t|t) \|^2 \! + \! \mu \! \left(\| \hat{\boldsymbol{v}}^*(t \! - \! M_t|t)\|^2 \right.} \notag \\
			& \qquad \qquad \qquad \qquad {\left.  \! +  \| \boldsymbol{v}(t \! - \! M_t)\|^2 \! + \! \| \boldsymbol{w}(t \! - \! M_t) \|^2 \right) \!,}\label{cont:1c}
		\end{align}
	\end{subequations}
	where $\boldsymbol{e}^*(t-j|t) = \boldsymbol{x}(t-j) - \hat{\boldsymbol{x}}^*(t-j|t)$, for all $j \in \mathbb{I}_{[0,M_t]}$.
	
	%	By iterating \eqref{cont:1}, it can be deduced that
	Substitute \eqref{cont:1b} into \eqref{cont:1a} and repeat this substitution for $\| \boldsymbol{e}^*(t-2|t) \|^2, \| \boldsymbol{e}^*(t-3|t) \|^2, \cdots, \| \boldsymbol{e}^*(t-M_t + 1|t) \|^2$, then
	\begin{align}
		%		\begin{aligned}
			&\| \boldsymbol{e}^*(t|t) \|^2  
			\leq   \rho^{2}\| \boldsymbol{e}^*(t \! - \! 2|t) \|^2 \notag \\
			& \qquad {+  \rho \mu \left( \| \hat{\boldsymbol{v}}^*(t \! - \! 2|t)\|^2
				\! + \! \| \boldsymbol{v}(t \! - \! 2)\|^2 \! + \! \| \boldsymbol{w}(t \! - \! 2) \|^2  \right)} \notag \\
			& \qquad { +  \mu \left( \| \hat{\boldsymbol{v}}^*(t \! - \! 1|t)\|^2
				\! + \! \| \boldsymbol{v}(t \! - \! 1)\|^2 \! + \! \| \boldsymbol{w}(t \! - \! 1) \|^2  \right)} \notag \\
			& {\leq  \cdots }
			\leq \rho^{M_t} \| \boldsymbol{e}^*(t \! - \! M_t|t) \|^2 \! + \! \sum_{j=1}^{M_t}\rho^{j-1} \mu \left( \| \hat{\boldsymbol{v}}^*(t \! - \! j|t)\|^2 \right. \notag \\
			& \qquad \qquad  \qquad \qquad \left. +  \| \boldsymbol{v}(t \! - \! j)\|^2  + \| \boldsymbol{w}(t \! - \! j) \|^2  \right).
			%		\end{aligned}
	\end{align}
	
	Using the Cauchy-Schwarz inequality as 
	\begin{align}
		\label{Triangle ineq:2}
		\| \boldsymbol{e}^*(t-M_t|t) \|^2 \leq&  2\| \bar{\boldsymbol{e}}(t-M_t)\|^2 \notag  \\
		& + 2\| \hat{\boldsymbol{x}}^*(t-M_t|t) -  \bar{\boldsymbol{x}}(t-M_t) \|^2,
	\end{align}
	where $\bar{\boldsymbol{e}}(t-j) = \boldsymbol{x}(t-j) - \bar{\boldsymbol{x}}(t-j)$ for all $j \in \mathbb{I}_{[1,M_t]}$, 
	we can continue the derivation to obtain that
	\begin{align}\label{Part I}
		%		\begin{aligned}
			\| \boldsymbol{e}^*(t|t)\|^2 
			\overset{\eqref{Triangle ineq:2}}{\leq}& 2\rho^{M_t}\|\bar{\boldsymbol{e}}(t-M_t) \|^2 \notag \\
			&+ \sum_{j=1}^{M_t}\rho^{j-1} \mu( \| \boldsymbol{w}(t-j) \|^2 + \| \boldsymbol{v}(t-j)\|^2) \notag \\
			&+ 2\rho^{M_t}\| \hat{\boldsymbol{x}}^*(t-M_t|t) - \bar{\boldsymbol{x}}(t-M_t) \|^2 \notag \\
			&+ \sum_{j=1}^{M_t}\rho^{j-1} \mu \| \hat{\boldsymbol{v}}^*(t-j|t)\|^2 \notag  \\
			\overset{\eqref{Eq:OBJ}}{\leq}& 2\rho^{M_t}\| \boldsymbol{x}(t-M_t) - \bar{\boldsymbol{x}}(t-M_t) \|^2 \notag  \\
			&+ \sum_{j=1}^{M_t}\rho^{j-1} \mu( \| \boldsymbol{w}(t-j) \|^2 + \| \boldsymbol{v}(t-j)\|^2) \notag \\
			&+ J_{L}(\hat{\boldsymbol{x}}^*(t-M_t|t),\hat{\boldsymbol{v}}^*(\cdot|t),t).
			%		\end{aligned}
	\end{align}
	
	\textbf{Part II:} Consider a sequence of state/output data $\boldsymbol{d}_{t} = \{[\boldsymbol{x}(t-M_t), \boldsymbol{x}(t-M_t+1)-\boldsymbol{w}(t-M_t),\dots,\boldsymbol{x}(t)-\boldsymbol{w}(t-1)],[\boldsymbol{y}(t-M_t)-\boldsymbol{v}(t-M_t),\dots,\boldsymbol{y}(t-1)-\boldsymbol{v}(t-1)] \}$. It is obvious that $\boldsymbol{d}_{t}$ is a disturbance-free and noise-free trajectory of the system \eqref{Eq:LDAS} and satisfies \eqref{opti1:1}. Combined with \textit{Assumption \ref{assum2}}, the data sequence $\boldsymbol{d}_{t}$ is a feasible solution to the optimization problem \eqref{opti1}. Thus, we have
	\begin{align}
		\label{OBJ comp}
		J_{L}(\hat{\boldsymbol{x}}^*(t-M_t|t),\hat{\boldsymbol{v}}^*(\cdot|t),t) \leq J_{L}({\boldsymbol{x}}(t-M_t),\boldsymbol{v}_{[t-M_t,t-1]},t).
	\end{align}
	
	Substituting \eqref{OBJ comp} into \eqref{Part I} yields that
	\begin{align}\label{proof:M-step}
		%	\begin{aligned}
			\| {\boldsymbol{e}}^*(t|t) \|^2 \leq& 2\rho^{M_t}\| \boldsymbol{x}(t-M_t) - \bar{\boldsymbol{x}}(t-M_t) \|^2 \notag \\
			&+ \sum_{j=1}^{M_t}\rho^{j-1} \mu( \| \boldsymbol{w}(t-j) \|^2 + \| \boldsymbol{v}(t-j)\|^2) \notag \\
			&+ J_{L}(\boldsymbol{x}(t-M_t),\boldsymbol{v}_{[t-M_t,t-1]},t) \notag \\
			= & 4\rho^{M_t}\| \boldsymbol{x}(t - M_t) - \bar{\boldsymbol{x}}(t-M_t) \|^2 \notag \\
			& + \! \sum_{j=1}^{M_t}\rho^{j-1}  \mu  ( \| \boldsymbol{w}(t \! - \! j) \|^2 \! + \! 2\| \boldsymbol{v}(t \! - \! j)\|^2).
			%	\end{aligned}
	\end{align}
	
	It completes the proof. $\Box$
\end{proof}

From the above proof, it is clear that the contraction property of the $M$-step Lyapunov function is satisfied by
\begin{align}
	\label{def:kappa}
	\kappa^{M} = 4\rho^{M} < 1.
\end{align}
{Note that since we assume observability, the robust stability of MHE still holds for $\rho = 0$ and $\kappa = 0$. The interested reader is referred to Section 4.3.1 of \cite{rawlings2017model} to complete the proof of robust stability for $\rho = 0$ and $\kappa = 0$.}

Finally, we give RGES of Dd-MHE based on $M$-step Lyapunov function as follows.
\begin{theorem}\label{theorem2}
	Let \textit{Assumptions \ref{assum:1}} and \textit{\ref{assum2}} hold, and the horizon length $M$ satisfies \eqref{def:kappa}. Then, for all $t \in \mathbb{I}_{\geq 0}$, the estimation error $\boldsymbol{x}(t) - \hat{\boldsymbol{x}}(t)$ satisfies
	\begin{align}\label{theorem2:eq}
		&\| \boldsymbol{x}(t) - \hat{\boldsymbol{x}}(t) \|\leq (\sqrt{\kappa})^{t}\| \boldsymbol{x}(0) - \bar{\boldsymbol{x}}(0) \| \notag \\
		&\qquad \qquad + \sum_{j=0}^{t-1}(\sqrt{\kappa})^{t-1-j}\sqrt{\mu}( \| \boldsymbol{w}(j) \| + \sqrt{2}\| \boldsymbol{v}(j)\|),
	\end{align}
	with $\kappa \in (0,1)$ as defined in \eqref{def:kappa} and $\rho, \mu$ in \eqref{def:rho mu}. Thus, Dd-MHE is RGES. 
\end{theorem}
\begin{proof}
	This proof has similar steps to the proof of Corollary 1 in \cite{schiller2023lyapunov}. We only list the key steps here.
	
	Let $t = kM+l$, where $l \in \mathbb{I}_{[1,M-1]}$ and $k \in \mathbb{I}_{\geq 0}$. Thus, when $k=0$, we have $ l = M_t =t$. Then, according to \eqref{proof:M-step} and $0 < \rho <\kappa $, there is
	\begin{align}\label{RGES:eq1}
		%		\begin{aligned}
			& \|{\boldsymbol{e}}(l)\|^2 = \| {\boldsymbol{e}}^*(l|l) \|^2 \leq \kappa^{l}\| \boldsymbol{x}(0) - \bar{\boldsymbol{x}}(0) \|^2 \notag \\
			&\qquad \qquad + \sum_{j=1}^{l}\kappa^{j-1}\mu( \| \boldsymbol{w}(l-j) \|^2 + 2\| \boldsymbol{v}(l-j)\|^2).
			%		\end{aligned}
	\end{align}

	For $k \geq 1$, we have $M_t =M$. Through \eqref{theorem:eq 1}, we can get
	\begin{align}\label{RGES:eq2}
		%	\begin{aligned}
			\| {\boldsymbol{e}}(t) \|^2 \leq&  \kappa^{kM} \| \boldsymbol{x}(l) - \bar{\boldsymbol{x}}(l) \|^2 \notag \\
			&+ \sum_{i=0}^{k-1} \kappa^{iM}\sum_{j=1}^{M}\kappa^{j-1}\mu \left( \| \boldsymbol{w}(t-iM-j) \|^2 \right.  \notag \\
			&\qquad \qquad \qquad \qquad \left. + 2\| \boldsymbol{v}(t-iM-j)\|^2 \right).
			%	\end{aligned}
	\end{align}
	
	Substituting \eqref{def:kappa} and \eqref{RGES:eq1} into \eqref{RGES:eq2}, the following equation can be given
	\begin{align}\label{RGES:eq3}
		%	\begin{aligned}
			&\| {\boldsymbol{e}}(t)\|^2 \leq  \kappa^{t}\| x(0) - \bar{\boldsymbol{x}}(0) \|^2 \notag \\
			&\! +\!  \sum_{j=1}^{l}\kappa^{kM+j-1}\mu( \| \boldsymbol{w}(t \! - \! kM \! - \! j) \|^2 + 2\| \boldsymbol{v}(t \! - \! kM \! - \! j)\|^2) \notag \\
			&\! +\! \sum_{i=0}^{k-1}\!{ \sum_{j=1}^{M} \kappa^{iM+j-1}\mu( \| \boldsymbol{w}(t \! - \! iM \! - \! j) \|^2 \! + \! 2\| \boldsymbol{v}(t \! - \! iM \! - \!j)\|^2) } \notag \\
			&\!=\! \kappa^{t}\| \boldsymbol{x}(0) \! - \! \bar{\boldsymbol{x}}(0) \|^2 \! + \!  \sum_{q=0}^{t-1}\kappa^{t-1-q}\mu( \| \boldsymbol{w}(q) \|^2 \! + \! 2\| \boldsymbol{v}(q)\|^2).
			%	\end{aligned}
	\end{align}
	
	Combining \eqref{RGES:eq1} and \eqref{RGES:eq3}, and using the fact that $\sqrt{\sum_{i=1}^{n} a_i } \leq \sum_{i=1}^{n} \sqrt{a_i}$ for all $a_i \geq 0$, we obtain
	\begin{align}\label{RGES:eq4}
		%	\begin{aligned}
			\| {\boldsymbol{e}}(t)\| \leq& (\sqrt{\kappa})^{t}\| \boldsymbol{x}(0) - \bar{\boldsymbol{x}}(0) \| \notag \\
			&+ \sum_{q=0}^{t-1}(\sqrt{\kappa})^{t-1-q}\sqrt{\mu}( \| \boldsymbol{w}(q) \| + \sqrt{2}\| \boldsymbol{v}(q)\|),
			%	\end{aligned}
	\end{align}
	where $t \in \mathbb{I}_{\geq 0}$.  $\Box$
\end{proof}

\section{Experiment and simulation}
\label{sec:4}
In this section, we design an experimental system and a full 3 Degrees Of Freedom (3-DOF) rotational spacecraft dynamics simulation to verify \textbf{Algorithm \ref{algorithm1}} on the eddy current de-tumbling mission. The Data-driven Koopman MHE (Dd-KMHE) and Data-driven Error State Kalman Filter (Dd-ESKF) are used as a comparison. The basic structure of the Dd-KMHE is shown in \cite{surana2020koopman}, where the Extended Dynamic Mode Decomposition (EDMD) algorithm based on thin plate spline radial basis functions is used to generate the lifted linear system \cite{korda2018linear}. {The basic framework of the Dd-ESKF is consistent with Section II-B of \cite{madyastha2011extended}, where the system model is given by \textit{Theorem 1} of \cite{de2019formulas}.}
The whole experiment includes nominal experiments, disturbance and noise experiments. The physical parameters in the 3-DOF rotational dynamics simulation are consistent with \cite{liu2022robust}.

\subsection{Eddy Current De-tumbling Experiment}
\subsubsection{Experiment description}
\label{subsec:4.1}
As shown in Fig.\ref{fig:1}, the experimental system includes a 2-axis linear module, a NdFeB permanent magnet, a target rotation system, and a computer. The 2-axis linear module adjusts the relative distance between the NdFeB permanent magnet and the target. The NdFeB permanent magnet generates a magnetic field, and a Net F/T Transducer is mounted underneath it to measure the de-tumbling torque.

\begin{figure}[htbp]\centering
	\includegraphics[width = 8.5cm]{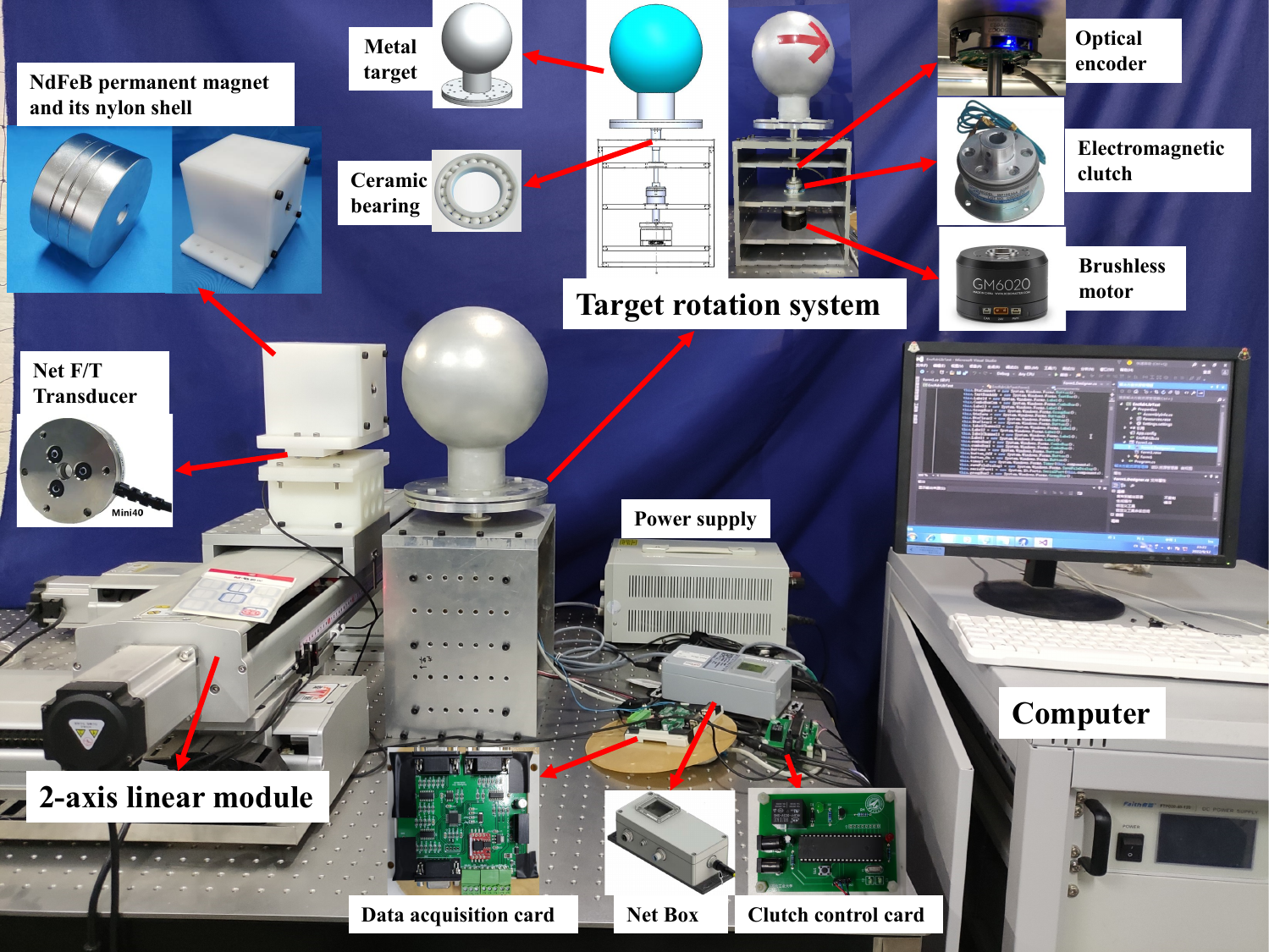}
	\caption{Eddy current de-tumbling experimental system}\label{fig:1}
\end{figure}

The target rotation system is also shown in Fig.\ref{fig:1} where a metal target starts to rotate under the drive of the brushless motor. After reaching a certain rotation speed, the computer sends a command to disengage the electromagnetic clutch, and the target starts to rotate freely. During this period, an optical encoder measures the target's angular velocity and sends the target's angular velocity information to the computer through the data acquisition card. Moreover, a ceramic bearing is used to reduce the frictional force during the target rotation.

The parameters of Dd-MHE in the experiment are $\rho = 0.9$, $M = 20$, $\mu = 10^5$. The constraint sets $\mathbb{V} = \mathbb{R}$ and $\mathbb{X} = \{\hat{\boldsymbol{x}}(t) | \underline{\boldsymbol{x}} \leq \hat{\boldsymbol{x}}(t) - \boldsymbol{x}_c(t) \leq \bar{\boldsymbol{x}} \}$ where $\boldsymbol{x}_c(t) = \boldsymbol{H}_{M_t+1}({\boldsymbol{x}}_{[-T, 0]}) (\boldsymbol{H}_{M_t}(\boldsymbol{y}_{[-T, -1]}))^{\dagger}\boldsymbol{y}(t)$, $\bar{\boldsymbol{x}} = - \underline{\boldsymbol{x}} = \frac{10\pi}{60} \boldsymbol{1}_{3}$. For comparison, Dd-KMHE adopts the same $\rho$, $M$, and $\mu$, and the lifting dimension is set to 20. The center of the thin plate spline radial basis function is generated by random sequence. The sampling interval of data during the experiment is $T_s = 0.01 \ \rm{(s)}$. The data in the first $0.5 \ \rm{(s)}$, a total of 50 data, are used as historical data, so the estimated value in the first $0.5 \ \rm{(s)}$ will be equal to the true value.

\subsubsection{Experimental results and analysis}
In this section, we show and analyze the results of the nominal experiment, the large sensor noise experiment, the large process disturbance experiment, and the experiment that contains both large noise and large disturbance. 

Fig.\ref{fig:2} shows that Dd-MHE, Dd-KMHE and Dd-ESKF converge for the nominal experiment. The parameters of Dd-ESKF are selected as $Q=10^5$ and $R=10^3$. Such a choice of parameters means that Dd-ESKF will rely more on the system model than on the measured output. Therefore, Dd-ESKF exhibits strong noise robustness, but cannot track the state changes caused by disturbances, as shown in Fig.\ref{fig:3} and Fig.\ref{fig:4}.

\begin{figure}[htbp]\centering
	\includegraphics[scale=0.5]{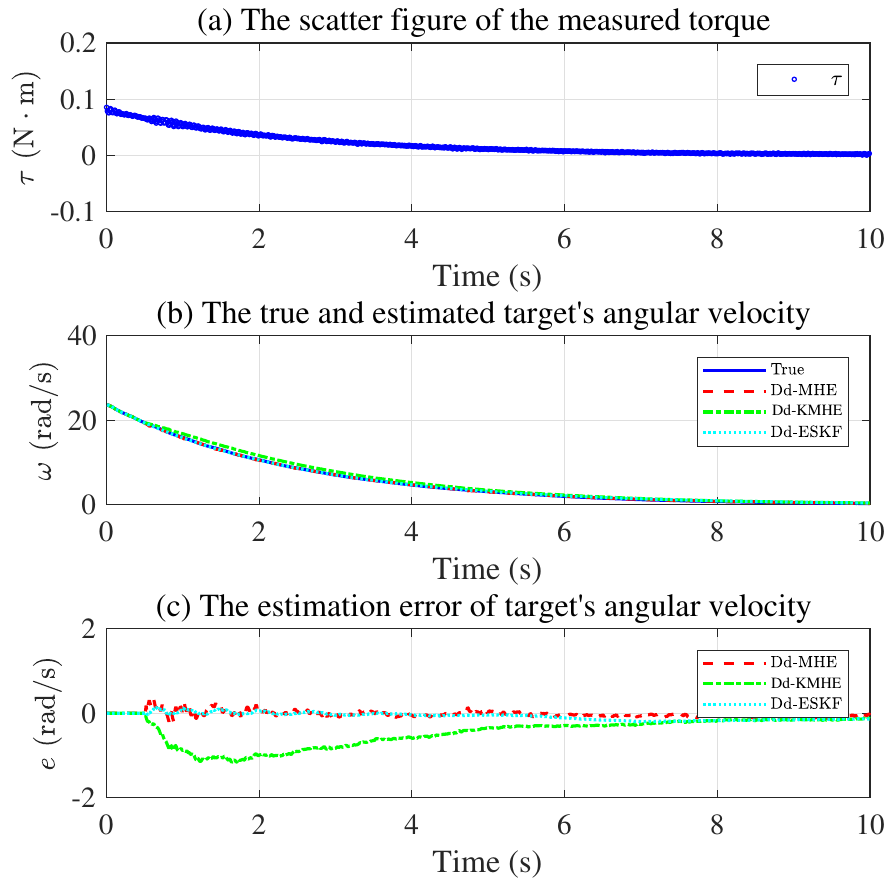}
	\caption{Results of nominal estimation experiment.}\label{fig:2}
\end{figure}

In the noise experiment (Fig.\ref{fig:3}), the estimation results of both Dd-MHE and Dd-KMHE show a bias, but both converge to the true value quickly after the noise ends. This bias is acceptable based on the fact that the estimators proposed in this paper are both data-driven. {In addition, from Fig.\ref{fig:3} to Fig.\ref{fig:5}, it can be found that Dd-MHE is more noise contaminated than that of the Dd-KMHE. It may be related to the fact that we only consider noise in the MHE optimization problem. As mentioned in \textit{Remark \ref{remark 1}}, the future focus will be on this issue.}

\begin{figure}[htbp]\centering
	\includegraphics[scale=0.5]{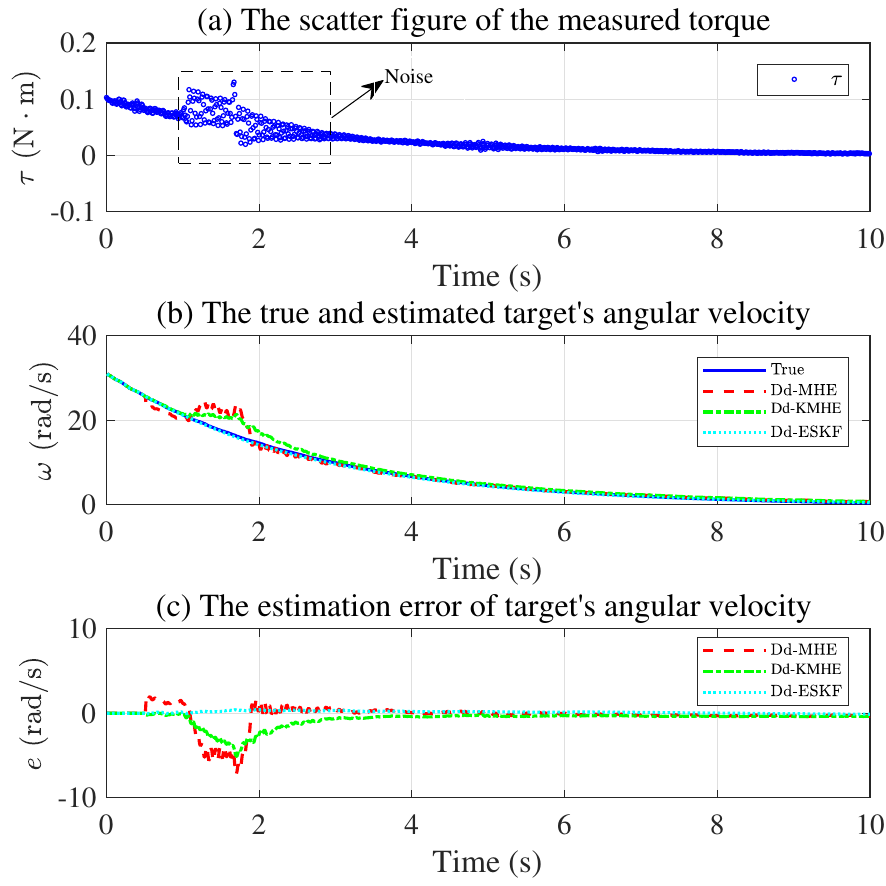}
	\caption{Results of noise estimation experiment.}\label{fig:3}
\end{figure}

\begin{figure}[htbp]\centering
	\includegraphics[scale=0.5]{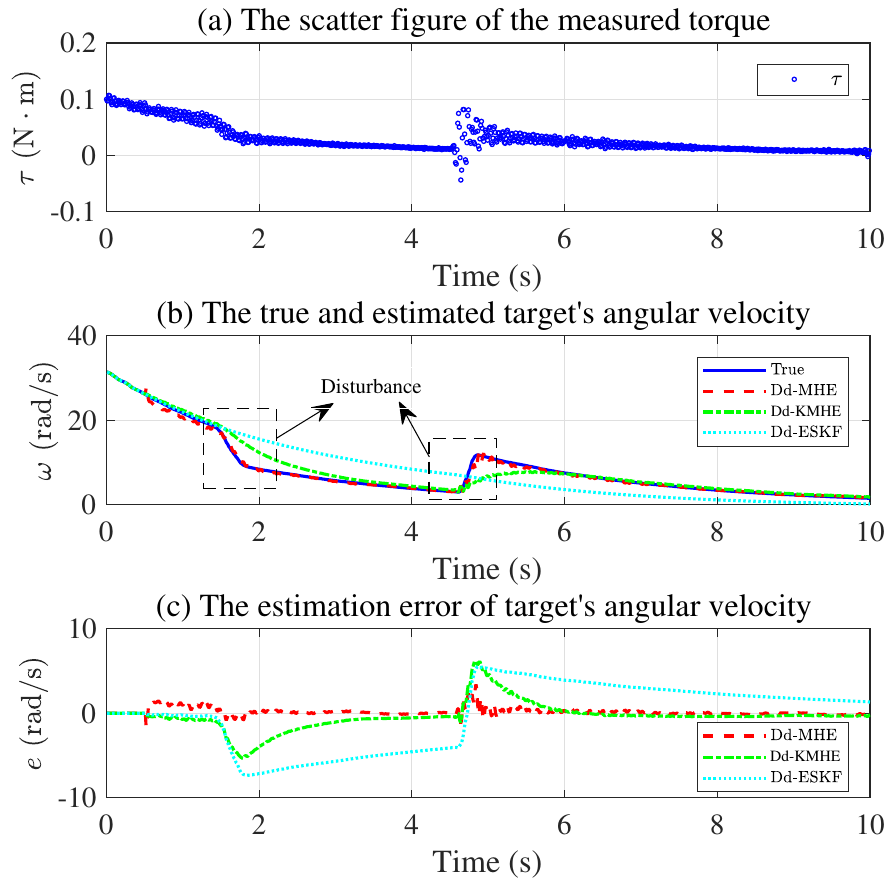}
	\caption{Results of disturbance estimation experiment.}\label{fig:4}
\end{figure}

The disturbance experiment (Fig.\ref{fig:4}) focuses on exploring the adaptability of the estimator to large process disturbances. Clearly, Dd-MHE shows significantly superior disturbance adaptation. It quickly tracks the true value after the disturbance. Dd-KMHE is less adaptive to the disturbance compared to Dd-MHE. In addition, we explored the estimator's performance for the case of containing both noise and disturbances. As shown in Fig.\ref{fig:5}, Dd-KMHE exhibits unacceptably large deviations, while Dd-MHE exhibits good performance under complex perturbations and noise. {The reason for the large performance difference is the lack of accuracy of the prediction model in Dd-KMHE due to the small amount of data.}

\begin{figure}[htbp]\centering
	\includegraphics[scale=0.5]{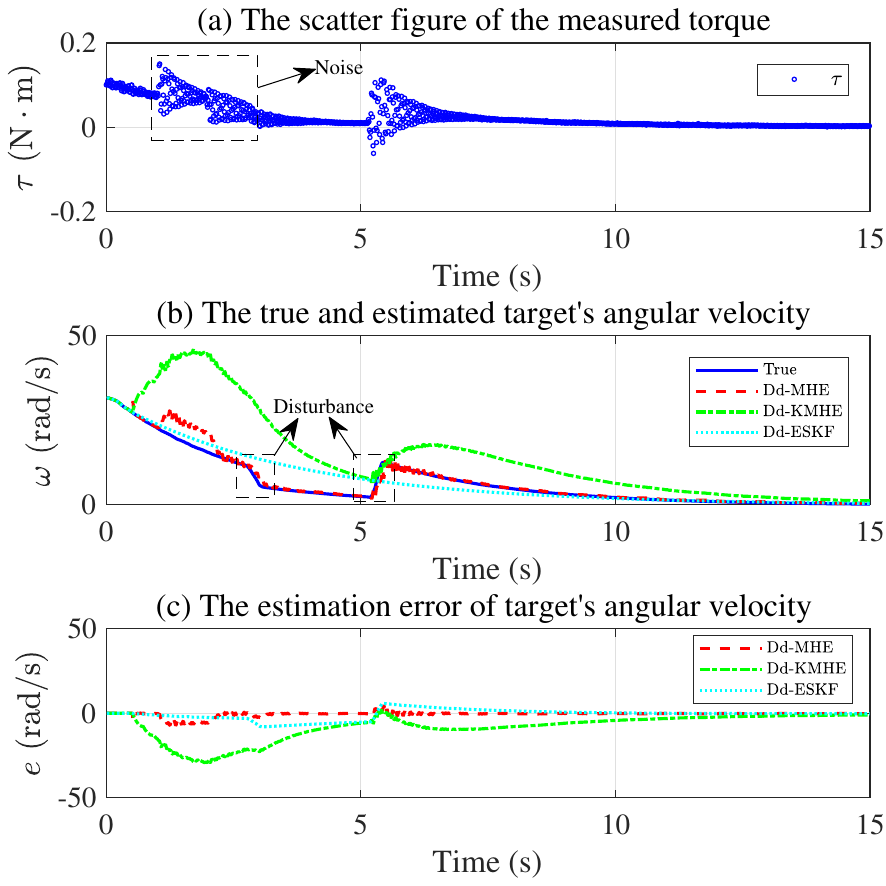}
	\caption{Results of noise and disturbance estimation experiment.}\label{fig:5}
\end{figure}

Finally, we performed the sensitivity analysis for the parameter $\mu$. This parameter is essential for the performance and stability of the estimator. According to \textbf{Algorithm \ref{algorithm2}}, $\rho_0 = 1.2705 \times 10^{-19}$, $\mu_0 = 3.3307 \times 10^{-14}$. 
Fig.\ref{fig:6} shows the estimation results of Dd-MHE and Dd-KMHE for different values of $\mu$. Although the estimation results are stable for all $\mu$, Dd-KMHE is much more sensitive to $\mu$, while Dd-MHE is less affected by $\mu$. It means that when we do not know much about the system to be estimated, an unreasonable choice for $\mu$ can lead to unacceptable bias in Dd-KMHE, while Dd-MHE tends to have a better estimation performance.
\begin{figure}[htbp]\centering
	\includegraphics[scale=0.5]{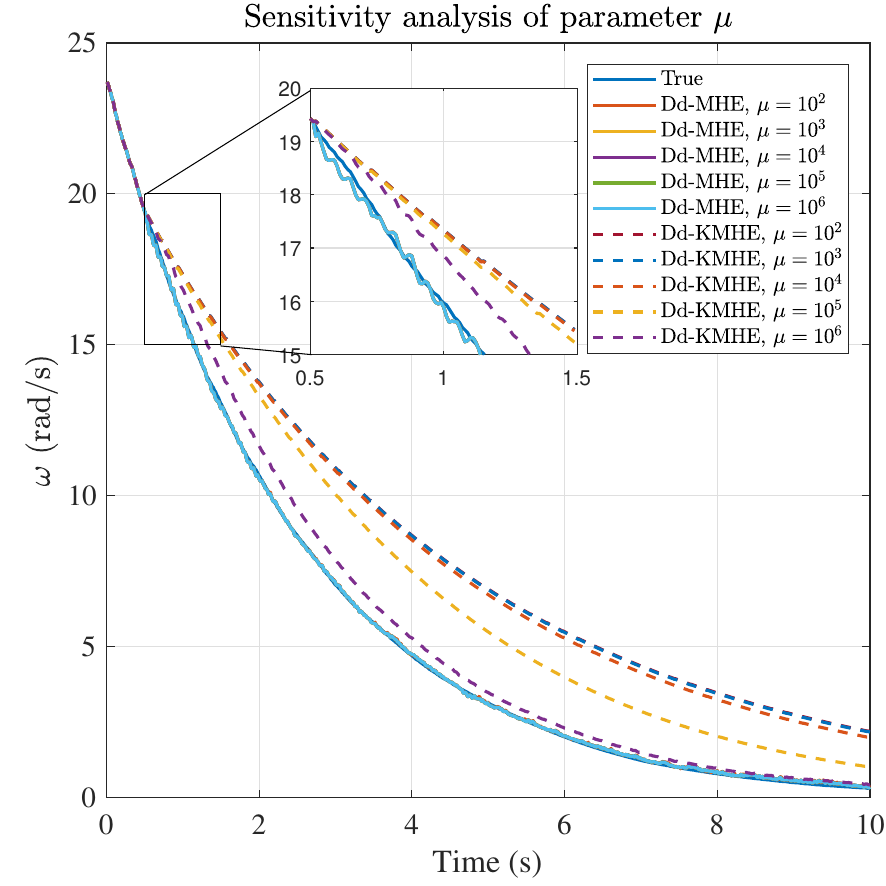}
	\caption{Sensitivity analysis of parameter $\mu$.}\label{fig:6}
\end{figure}

\subsection{{3-DOF rotational dynamics simulation}}
In this section, 3-DOF rotational dynamics simulation is adopted to demonstrate the performance of Dd-MHE for nonlinear systems.
In this simulation, we employ the magnetic dipole approximation model for the de-tumbling torque, which is given as \cite{gomez2017guidance}:
\begin{align}
	\boldsymbol{\tau}_{t}(\boldsymbol{\omega}_t) = \left(\boldsymbol{M}_{\rm{eff}} \left( \left( \boldsymbol{\omega}_t - \boldsymbol{\omega}_c \right)  \times \boldsymbol{B}_{Gt}\right)   \right) \times \boldsymbol{B}_{Gt},
\end{align} 
where $\boldsymbol{M}_{\rm{eff}}$ denotes the effective magnetic tensor, which is a constant matrix for a space target. $\boldsymbol{\omega}_c$ denotes the angular velocity of the chaser, assumed to be constant. $\boldsymbol{B}_{Gt}$ is the magnetic field at the center of gravity (COG) of the target. And the measurement output is
\begin{align}
	\boldsymbol{\tau}_{c}(\boldsymbol{\omega}_t) = -\boldsymbol{\tau}_{t}(\boldsymbol{\omega}_t) - \boldsymbol{r} \times \boldsymbol{F}_{ct}(\boldsymbol{\omega}_t),
\end{align}
where $\boldsymbol{r}$ represents the relative position from the chaser to the target, $\boldsymbol{F}_{ct}(\boldsymbol{\omega}_t)$ is the induced  force on the target by the chaser given by
\begin{align}
	\boldsymbol{F}_{ct}(\boldsymbol{\omega}_t) = \boldsymbol{\Lambda}_{Gt} \boldsymbol{M}_{\rm{eff}}( \left( \boldsymbol{\omega}_t - \boldsymbol{\omega}_c \right)  \times \boldsymbol{B}_{Gt} ),
\end{align}
where $\boldsymbol{\Lambda}_{Gt}$ is the Jacobian tensor of the magnetic field at the COG of the target.
The above parameters in this simulation are taken as: $\boldsymbol{J}_t = \rm{diag}([4513.2,4138.1,3282.5]) \ \rm{(kg \cdot m^2)}$, $\boldsymbol{M}_{\rm{eff}} = 0.89 \cdot 10^5 \cdot \rm{diag}([5.908,5.908,1.951]) \ \rm{(S \cdot m^4)}$, $\boldsymbol{\omega}_c = \boldsymbol{0} \ \rm{(deg/s)}$, $\boldsymbol{\omega}_t = [14.364,1.224,3.4195]^\top \ \rm{(deg/s)}$, $\boldsymbol{B}_{Gt} = [41,51,-41]^\top \cdot 10^{-4} \ \rm{(T)}$, $\boldsymbol{r} = [0.5,1,0.5]^\top \rm{(m)}$, and
$$\boldsymbol{\Lambda}_{Gt} = \begin{bmatrix}
	-23 & -116  & 8 \\
	-116 & -119 & 49 \\
	8  & 49  & 142
\end{bmatrix}\cdot 10^{-4}.$$

\begin{figure}[htbp]\centering
	\includegraphics[width = 8.5 cm]{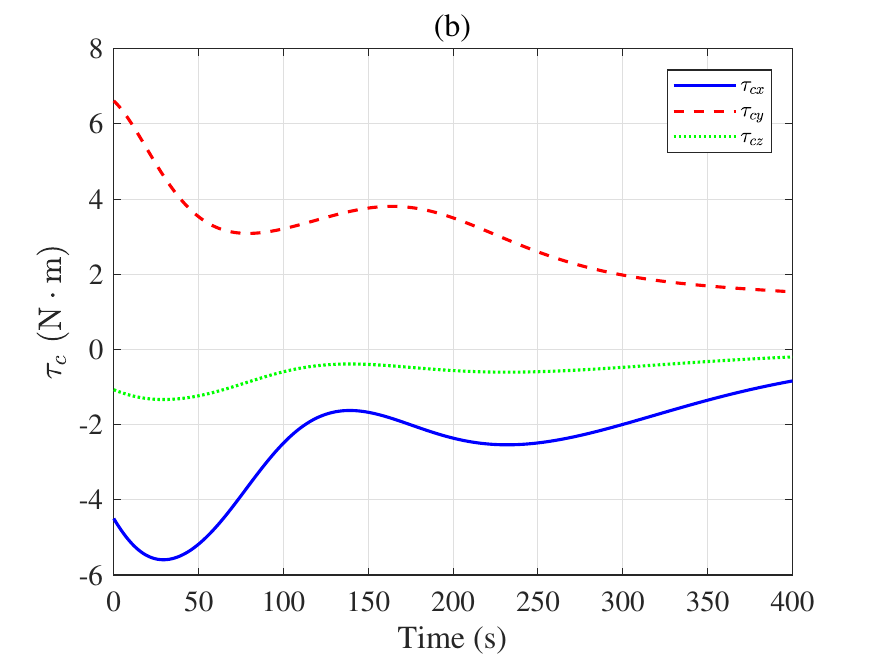}
	\caption{Output of the simulation case}\label{fig:7}
\end{figure}

The parameters of Dd-MHE are: $\rho = 0.8$, $M = 10$, $\mu = 10^{5}$ with $\rho_0 = 0.4614$ and $\mu_0 = 21.1561$. The output of the simulation case is shown in Fig.\ref{fig:7}. In this simulation, we assume that the measured values in the first $20 \  \rm{(s)}$ are known and use them as historical data. 
%Therefore, in the first $20 \  \rm{(s)}$ of Fig.\ref{fig:5}, the truth value of the target angular velocity is equal to the estimated value of the two algorithms. 
After the $20 \  \rm{(s)}$, the target's angular velocity can no longer be obtained directly, and the estimation of the target's angular velocity begins. 
%The true angular velocity of the target is shown by the solid blue line in Fig.\ref{fig:8}. The results of target angular velocity estimation using the data-driven MHE algorithm proposed in this paper are shown in the red dotted line in Fig.\ref{fig:5}. For comparison, the estimation results of data-driven MHE based on Koopman operator are shown in the green dotted line in Fig.\ref{fig:5}.
\begin{figure}[htbp]\centering
	\includegraphics[width = 8.5cm]{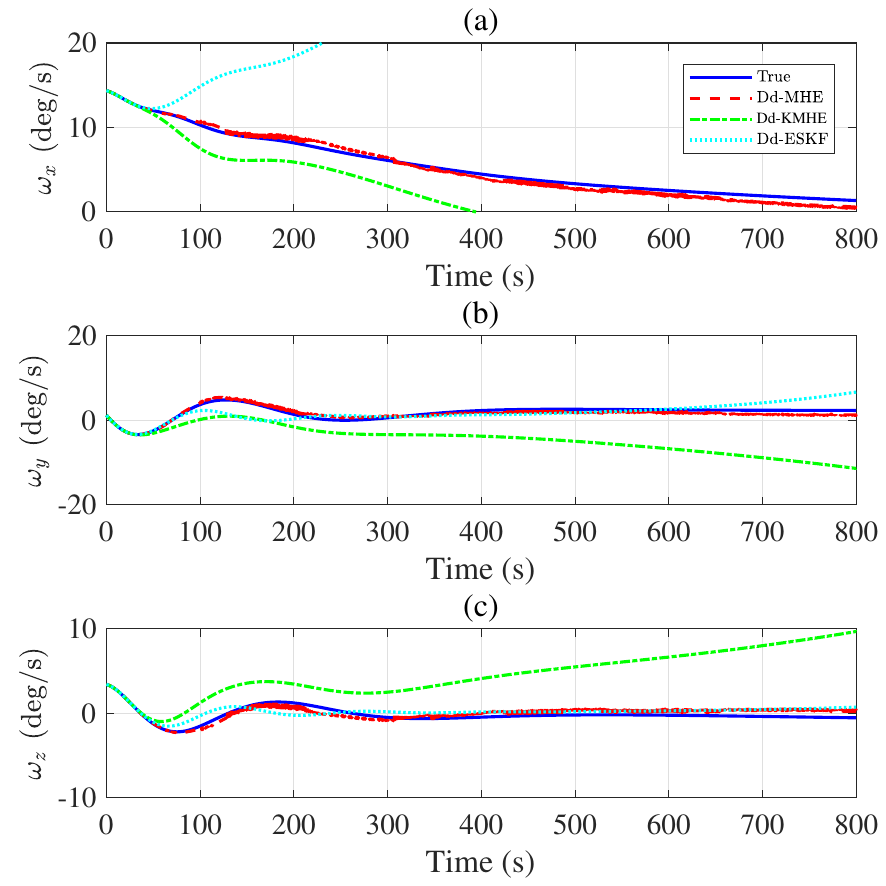}
	\caption{The target's angular velocity and its estimation of 3-DOF rotational dynamics simulation}\label{fig:8}
\end{figure}

It can be seen from Fig.\ref{fig:8} that Dd-MHE can well estimate the angular velocity of the target, but the performance of the Dd-KMHE and Dd-ESKF declines very fast. In the whole simulation, the estimated value by Dd-MHE is stable, while the estimated value by Dd-KMHE and Dd-ESKF is divergent. The reason for the estimated value by Dd-KMHE and Dd-ESKF diverging is that they both need a lot of data to obtain an accurate approximate model or appropriate parameters to reduce the dependence on the prediction model. In contrast, the Willems' fundamental lemma only needs to satisfy the rank condition \eqref{rank:H1}. Therefore, Dd-MHE can estimate the angular velocity of the target well of a small data set. 

\section{Conclusion}
This paper proposes a data-driven moving horizon estimation method for small data sets and applies it to the angular velocity estimation of noncooperative targets in the eddy current de-tumbling mission. Compared with data-driven methods for large data sets or appropriate parameters, such as Dd-KMHE and Dd-ESKF, the proposed method in this paper has better estimation performance and is less sensitive to the parameter.
However, since the method proposed in this paper is based on a local linear approximation, its effectiveness in strongly nonlinear systems is not yet known. It will be an issue that needs further study in the future.

\appendices

\section{Pole Placement for Unknown Systems}
\label{sec:Appendix A}
The appendix is intended to provide a pole placement method for unknown systems. The method can give a state observer gain $\boldsymbol{L}$ of unknown systems and make \eqref{constraint rho mu} hold.

\subsection{Duel system for Data-Driven Estimation}\label{sec:Appendix A1}
Consider the following unknown linear observable system
\begin{align}
	\label{Eq:ULAS}
	\begin{aligned}
		\boldsymbol{x}(t+1) = \boldsymbol{A} \boldsymbol{x}(t), \quad \boldsymbol{y}(t) = \boldsymbol{C} \boldsymbol{x}(t),
	\end{aligned}
\end{align}
where $\boldsymbol{x}(t) \in \mathbb{R}^{n}$ and $\boldsymbol{y}(t) \in \mathbb{R}^{p}$ denote the system state and output at time $t$, respectively. $\boldsymbol{A} \in \mathbb{R}^{n \times n}$ is the system matrix, $\boldsymbol{C} \in \mathbb{R}^{p \times n}$ is the output matrix, and they are both unknown. $n$ and $p$ are the dimensions of the system state and output, respectively.

The dual system \cite{callier1991linear} for \eqref{Eq:ULAS} is defined as
\begin{align}\label{Eq:dual system}
	\boldsymbol{x}_{dl}(t+1) = \boldsymbol{A}^\top \boldsymbol{x}_{dl}(t) + \boldsymbol{C}^\top \boldsymbol{u}_{dl}(t),
\end{align}
where $\boldsymbol{x}_{dl}(t) \in \mathbb{R}^{n}$ and $\boldsymbol{u}_{dl}(t) \in \mathbb{R}^{p}$ denote the state and input of dual system at time $t$, respectively.
%Obviously, the controller pole placement of the dual system \eqref{Eq:dual system} is equivalent to the observer pole placement of the original system \eqref{Eq:ULAS}. 
Obviously, the observability of the original system \eqref{Eq:ULAS} is equivalent to the controllability of the dual system \eqref{Eq:dual system}.  Then, our purpose is transformed into designing a controller gain matrix $\boldsymbol{L}^\top$ of the dual system to satisfy \eqref{constraint rho mu}.
A detailed discussion of this idea can be found in \cite{adachi2021dual}.

%Therefore, if the feedback control law of the dual system is taken as $\boldsymbol{u}_{dl}(t) = -\boldsymbol{L}^\top \boldsymbol{x}_{dl}(t)$, the poles of \eqref{Eq:dual system} are equal to the poles of \eqref{Eq:ULAS} with the state observer gain $\boldsymbol{L}$. Then, our purpose is transformed into designing a controller gain matrix $\boldsymbol{L}^\top$ of the dual system that satisfies \eqref{constraint rho mu}. A detailed discussion of this idea can be found in \cite{adachi2021dual}.

Considering that the trajectory of the dual system will be used later, 
%this subsection also gives a trajectory of the dual system, which is completely composed of a trajectory of the original system. To achieve this, 
we give the following assumption and lemma.
\begin{assumption}
	\label{assum:3}
	For a trajectory $\{ \boldsymbol{x}_{[0, T]}, \boldsymbol{y}_{[0, T-1]} \}$ of \eqref{Eq:ULAS} with length $T$, and let $\boldsymbol{X}_{0,T-1} =\boldsymbol{H}_1(\boldsymbol{x}_{[0,T-1]})$, $\boldsymbol{X}_{1,T} = \boldsymbol{H}_1(\boldsymbol{x}_{[1,T]})$, $\boldsymbol{Y}_{0,T-1} = \boldsymbol{H}_1(\boldsymbol{y}_{[0,T-1]})$, it is assumed that $\boldsymbol{X}_{0,T-1}$ and $\left[\boldsymbol{X}_{1,T}^\top \quad \boldsymbol{Y}_{0,T-1}^\top \right]$ have full rank, that is
	\begin{align}\label{output rank condition}
		\mathrm{rank} \! \left( \boldsymbol{X}_{0,T-1} \right) \! = \! n, \  \mathrm{rank} \! \left( \left[\boldsymbol{X}_{1,T}^\top \quad \boldsymbol{Y}_{0,T-1}^\top \right] \right) \! = \!  n \! + \! p.
	\end{align}
\end{assumption}

%Therefore, following the Assumption \ref{assum:3}, we have the proposition as:
\begin{lemma}\cite{adachi2021dual}\label{lemma2}
	The data sequence
	\begin{align}\label{trajectory of the dual system}
		\boldsymbol{X}^{dl}_{1,T} \! = \! (\boldsymbol{X}^\top_{0,T-1})^\dagger, 
		\begin{bmatrix}
			\boldsymbol{X}^{dl}_{0,T-1} \\ \boldsymbol{U}^{dl}_{0,T-1} 
		\end{bmatrix} \! = \!  \left[\boldsymbol{X}_{1,T}^\top \quad \boldsymbol{Y}_{0,T-1}^\top \right]^\dagger,
	\end{align}
	is a trajectory of the dual system \eqref{Eq:dual system} in the meaning of least squares.
\end{lemma}
\begin{remark}
	For the linear autonomous system with offsets such as \eqref{Eq:LAS}, let 
	\begin{subequations}\label{data for offsets system}
		\begin{align}
			&\boldsymbol{X}_{0,T-1} = \boldsymbol{H}_1(\boldsymbol{x}_{[1,T-1]}) - \boldsymbol{H}_1(\boldsymbol{x}_{[0,T-2]}), \\
			&\boldsymbol{X}_{1,T} =\boldsymbol{H}_1(\boldsymbol{x}_{[2,T]}) - \boldsymbol{H}_1(\boldsymbol{x}_{[1,T-1]}),\\
			&\boldsymbol{Y}_{0,T-1} = \boldsymbol{H}_1(\boldsymbol{y}_{[1,T-1]}) - \boldsymbol{H}_1(\boldsymbol{y}_{[0,T-2]}),
		\end{align}
	\end{subequations}
	%	$\boldsymbol{X}_{0,T-1} = \boldsymbol{H}_1(\boldsymbol{x}_{[1,T-1]}) - \boldsymbol{H}_1(\boldsymbol{x}_{[0,T-2]}) $, $\boldsymbol{X}_{1,T} =\boldsymbol{H}_1(\boldsymbol{x}_{[2,T]}) - \boldsymbol{H}_1(\boldsymbol{x}_{[1,T-1]}) $, $\boldsymbol{Y}_{0,T-1} = \boldsymbol{H}_1(\boldsymbol{y}_{[1,T-1]}) - \boldsymbol{H}_1(\boldsymbol{y}_{[0,T-2]})$, 
	then \textit{Lemma \ref{lemma2}} still holds under \textit{Assumption \ref{assum:3}}.  Note that in this case, the number of columns in $\boldsymbol{X}_{0,T-1}$ is $T-1$ (no longer $T$), so as $\boldsymbol{X}_{1,T}$ and $\boldsymbol{Y}_{0,T-1}$.
\end{remark}

\subsection{Pole Placement Method}
In subsection Appendix \ref{sec:Appendix A1}, we have transformed the pole placement for the state observer of the unknown system \eqref{Eq:ULAS} into that for the controller of the dual system \eqref{Eq:dual system}. The traditional pole placement method \cite{callier1991linear} includes a complex process, which is not conducive to unknown systems. In robust control, Linear Matrix Inequality (LMI) provides a convenient regional pole placement method \cite{chilali1996h}. In addition, the Linear Quadratic Regulation (LQR) of unknown systems are combined with LMI to ensure the controller's performance \cite{de2019formulas}. All these make LMI more suitable for dealing with the pole placement problem in this paper. 
Next, we briefly introduce the regional pole placement and data-driven LQR and give the pole placement method for unknown systems with guaranteed performance.
%Next, the method proposed in this paper will be described in detail.
\subsubsection{regional pole placement}
\begin{definition}\cite{chilali1996h}\label{Def:3}
	A subset $\mathcal{D}$ of the complex plane is called an LMI region if there exists a symmetric matrix $\boldsymbol{L}_{\mathcal{D}}$ and a matrix $\boldsymbol{M}_{\mathcal{D}}$ such that
	\begin{align}
		\mathcal{D} = \{ z \in \mathbb{C}: \boldsymbol{L}_{\mathcal{D}} + z \boldsymbol{M}_{\mathcal{D}} + z^* \boldsymbol{M}^\top_{\mathcal{D}} \prec 0 \},
	\end{align}
	where $\mathbb{C}$ stands for the sets of complex numbers, $z^*$ is the complex conjugate of $z$, ``$\prec 0$'' denotes negative definite. Correspondingly, ``$\succ 0$'' denotes positive definite. 
\end{definition}

Then, the following lemma gives the necessary and sufficient condition for the pole location in the LMI region $\mathcal{D}$.
\begin{lemma}\cite{chilali1996h}\label{lemma3}
	The pole of matrix $\boldsymbol{A}^\top_{\boldsymbol{L}}$ is located in the LMI region $\mathcal{D}$ if and only if there exists a symmetric positive definite matrix $\boldsymbol{X}_{\mathcal{D}}$ such that
	\begin{align}\label{Original pole placement}
		\boldsymbol{L}_{\mathcal{D}} \! \otimes \! \boldsymbol{X}_{\mathcal{D}} \! + \! \boldsymbol{M}_{\mathcal{D}} \! \otimes \! (\boldsymbol{A}^\top_{\boldsymbol{L}} \boldsymbol{X}_{\mathcal{D}}) \! +  \! \boldsymbol{M}_{\mathcal{D}}^\top \! \otimes \! (\boldsymbol{A}^\top_{\boldsymbol{L}} \boldsymbol{X}_{\mathcal{D}})^\top
		\! \succ \! 0, 
	\end{align}
	where $\boldsymbol{A}^\top_{\boldsymbol{L}} = \boldsymbol{A}^\top + \boldsymbol{C}^\top \boldsymbol{L}^\top$ is the system matrix of \eqref{Eq:dual system} with $\boldsymbol{u}_{dl}(t) = \boldsymbol{L}^\top \boldsymbol{x}_{dl}(t)$.
\end{lemma}
%\begin{remark}
%	For the constraint $\rho_0 := 6\lambda_{max}(\boldsymbol{A}_{\boldsymbol{L}}^\top \boldsymbol{A}_{\boldsymbol{L}}) \in (0,1)$, 

If the LMI region $\mathcal{D}$ is a circle with radius $r_{\mathcal{D}}$ centered at the origin, \eqref{Original pole placement} can be converted to\cite{chilali1996h,wisniewski2019discrete}
\begin{align}\label{circle LMI region}
	\begin{bmatrix}
		-r_{\mathcal{D}} \boldsymbol{X}_{\mathcal{D}} & \boldsymbol{A}^\top_{\boldsymbol{L}} \boldsymbol{X}_{\mathcal{D}} \\
		(\boldsymbol{A}^\top_{\boldsymbol{L}} \boldsymbol{X}_{\mathcal{D}})^\top & -r_{\mathcal{D}} \boldsymbol{X}_{\mathcal{D}}
	\end{bmatrix} \prec 0 .
\end{align}
%\end{remark}

\subsubsection{data-driven LQR}
For the dual system \eqref{Eq:dual system}, define the performance signal $\boldsymbol{z}_{dl}(t)$ as 
\begin{align}\label{performance signal}
	\boldsymbol{z}_{dl}(t) = 
	\begin{bmatrix}
		\boldsymbol{Q}_{dl}^{1/2}  & \boldsymbol{0} \\ \boldsymbol{0} & \boldsymbol{R}_{dl}^{1/2}
	\end{bmatrix} 
	\begin{bmatrix}
		\boldsymbol{x}_{dl}(t) \\  \boldsymbol{u}_{dl}(t)
	\end{bmatrix},
\end{align}
%	where $\boldsymbol{Q}_{dl}  \succ  0,  \boldsymbol{R}_{dl} \succ  0$ are weighting matrices.
where $\boldsymbol{Q}_{dl} \! = \! \boldsymbol{Q}_{dl}^\top \! \succ \! 0,  \boldsymbol{R}_{dl} \! = \! \boldsymbol{R}_{dl}^\top \! \succ \! 0$ are weighting matrices.
%	with $(\boldsymbol{Q}_{dl},\boldsymbol{A}^\top_{\boldsymbol{L}} )$ observable. 

Thus, we have the data-driven LQR lemma as:
\begin{lemma}\cite{de2019formulas}\label{lemma4}
	Let \textit{Assumption \ref{assum:3}} hold. Then, the data-driven LQR state-feedback gain $\boldsymbol{L}^\top$ for system \eqref{Eq:dual system} can be computed as $\boldsymbol{L}^\top = \boldsymbol{U}^{dl}_{0,T-1} \boldsymbol{X} (\boldsymbol{X}^{dl}_{0,T-1}\boldsymbol{X})^{-1}$ where $\boldsymbol{X}$ optimizes
	\begin{subequations}
		\begin{align}
			%			\begin{aligned}
				\min_{\boldsymbol{X},\boldsymbol{W}} \quad & J = \mathrm{Trace}(\boldsymbol{Q}_{dl} \boldsymbol{X}^{dl}_{0,T-1}\boldsymbol{X}) + \mathrm{Trace}( \boldsymbol{W} ), \\
				\label{performance}	\text{s.t.} \quad 
				& \begin{bmatrix}
					\boldsymbol{W} & \boldsymbol{R}_{dl}^{1/2} \boldsymbol{U}^{dl}_{0,T-1} \boldsymbol{X} \\
					(\boldsymbol{U}^{dl}_{0,T-1} \boldsymbol{X})^\top \boldsymbol{R}_{dl}^{1/2} & \boldsymbol{X}^{dl}_{0,T-1}\boldsymbol{X}
				\end{bmatrix}	\succ 0 ,\\
				\label{stability} &\begin{bmatrix}
					\boldsymbol{X}^{dl}_{0,T-1}\boldsymbol{X} - \boldsymbol{I}_n & \boldsymbol{X}^{dl}_{1,T}\boldsymbol{X} \\
					(\boldsymbol{X}^{dl}_{1,T}\boldsymbol{X})^\top & \boldsymbol{X}^{dl}_{0,T-1}\boldsymbol{X}
				\end{bmatrix} \succ 0.
				%		\end{aligned}
		\end{align}
	\end{subequations}
\end{lemma}

\subsubsection{Pole Placement for Unknown Systems}
Since $\boldsymbol{A}^\top_{\boldsymbol{L}}$ is unknown, \eqref{circle LMI region} needs to be converted into data-driven form. Inspired by \textit{Theorem 2} of \cite{de2019formulas}, we have that $\boldsymbol{A}^\top_{\boldsymbol{L}} = \boldsymbol{X}^{dl}_{1,T} \boldsymbol{G}_{L}$ where $\boldsymbol{G}_{L}$ is a $T \times n$ matrix satisfying 
\begin{align}\label{G_L}
	\begin{bmatrix}
		\boldsymbol{L}^\top \\
		\boldsymbol{I}_n 
	\end{bmatrix} = 	
	\begin{bmatrix}
		\boldsymbol{U}^{dl}_{0,T-1} \\
		\boldsymbol{X}^{dl}_{0,T-1}
	\end{bmatrix}\boldsymbol{G}_{L}.
\end{align} 
Let $\boldsymbol{X} = \boldsymbol{G}_{L}\boldsymbol{X}_{\mathcal{D}}$, it is easy to deduce that
\begin{align}\label{AL}
	\boldsymbol{A}^\top_{\boldsymbol{L}} = \boldsymbol{X}^{dl}_{1,T} \boldsymbol{X} \boldsymbol{X}^{-1}_{\mathcal{D}}, \quad \boldsymbol{X}_{\mathcal{D}} = \boldsymbol{X}^{dl}_{0,T-1}\boldsymbol{X}.
\end{align}
%$\boldsymbol{A}^\top_{\boldsymbol{L}} = \boldsymbol{X}^{dl}_{1,T} \boldsymbol{X} \boldsymbol{X}^{-1}_{\mathcal{D}} $ and $\boldsymbol{X}_{\mathcal{D}} = \boldsymbol{X}^{dl}_{0,T-1}\boldsymbol{X}$ . 
Thus, \eqref{circle LMI region} can be converted to
\begin{align}\label{dd circle LMI region}
	\begin{bmatrix}
		-r_{\mathcal{D}} \boldsymbol{X}^{dl}_{0,T-1}\boldsymbol{X} & \boldsymbol{X}^{dl}_{1,T} \boldsymbol{X} \\
		(\boldsymbol{X}^{dl}_{1,T} \boldsymbol{X})^\top & -r_{\mathcal{D}} \boldsymbol{X}^{dl}_{0,T-1}\boldsymbol{X}
	\end{bmatrix} \prec 0 .
\end{align}

In addition, considering that $\boldsymbol{X}_{\mathcal{D}}$ is a symmetric positive definite matrix, we give the following symmetric positive matrix constraints\cite{boukas2008control}:
\begin{subequations}
	\begin{align}
		%			\begin{aligned}
			\min_{\beta} \quad  & \beta \\
			\label{symmetry} \text{s.t.} \
			&\begin{bmatrix}
				-\beta \boldsymbol{I}_n &  \boldsymbol{X}^{dl}_{0,T-1}\boldsymbol{X} - \boldsymbol{X}^\top (\boldsymbol{X}^{dl}_{0,T-1})^\top \\
				\star & - \boldsymbol{I}_n
			\end{bmatrix} \prec 0 ,\\
			\label{positive} & \  \beta > 0, \  \boldsymbol{X}^{dl}_{0,T-1} \boldsymbol{X} \succ 0,
			%		\end{aligned}
	\end{align}
\end{subequations}	
where $\star $ is used to represent the symmetry of the matrix in \eqref{symmetry}. The symmetric matrix constraints can be found in section 10.1 of \cite{boukas2008control}.

To sum up, we give the following theorem to the observer pole placement for unknown systems.
\begin{theorem}\label{pole placement unknown systems}
	For the unknown linear observable system \eqref{Eq:ULAS}, let \textit{Assumption \ref{assum:3}} hold. Then, the observer pole is located in an LMI region $\mathcal{D}$ with performance signal \eqref{performance signal} if the observer gain $\boldsymbol{L}= ( \boldsymbol{U}^{dl}_{0,T-1} \boldsymbol{X} (\boldsymbol{X}^{dl}_{0,T-1}\boldsymbol{X})^{-1} )^\top $ where $\boldsymbol{X}$ optimizes
	\begin{subequations}\label{optimizes X}
		\begin{align}
			%			\begin{aligned}
				\min_{\boldsymbol{X},\boldsymbol{W},\beta} \  & J = \mathrm{Trace}(\boldsymbol{Q}_{dl} \boldsymbol{X}^{dl}_{0,T-1}\boldsymbol{X}) + \mathrm{Trace}( \boldsymbol{W} ) + \beta \\
				\text{s.t.} \ 
				&  \eqref{performance}, \eqref{stability}, \eqref{dd circle LMI region}, \eqref{symmetry} \ \mathrm{and} \ \eqref{positive},
				%		\end{aligned}
		\end{align}
	\end{subequations}	
	where $\mathcal{D}$ is a circle with radius $r_{\mathcal{D}}$ centered at the origin.
\end{theorem}

Finally, the following algorithm is provided to satisfy \eqref{constraint rho mu}.
%Finally, the following algorithm is given to make \eqref{constraint rho mu} satisfy.
\begin{algorithm}[]  %
	\label{algorithm2}
	\caption{Pole Placement for the unknown system \eqref{Eq:LAS} to satisfy \eqref{constraint rho mu}.}%算法名字
	%	\LinesNumbered %
	\textbf{Data Preprocessing:} Given a trajectory of the system \eqref{Eq:LAS} as $\{ \boldsymbol{x}_{[0, T]}, \boldsymbol{y}_{[0, T-1]} \}$; Given the initial guess $\rho_0  \geq 1$ and $r_{\mathcal{D}} < 1 $; Given the symmetric weighting matrices $\boldsymbol{Q}_{dl}$, $\boldsymbol{R}_{dl}$; Define $\boldsymbol{X}_{0,T-1}$, $\boldsymbol{X}_{1,T}$, $\boldsymbol{Y}_{0,T-1}$ as \eqref{data for offsets system}; Define $\boldsymbol{X}^{dl}_{0,T-1}$, $\boldsymbol{X}^{dl}_{1,T}$, $\boldsymbol{U}^{dl}_{0,T-1}$ as \eqref{trajectory of the dual system}.\\ %\;用于换行
	\textbf{Numerical Iteration:}\\
	%\begin{algorithmic}[1]
	\eIf{\eqref{output rank condition} holds}{
		\While{$\rho_0 \geq 1$}{
			Solve equation \eqref{optimizes X} and obtain $\boldsymbol{X}$\;
			Let $\rho_0 \! = \! 6\lambda_{max}( \! \boldsymbol{A}_{\boldsymbol{L}}^\top \! \boldsymbol{A}_{\boldsymbol{L}})$ and 
			$\mu_0 \! = \! 6\lambda_{max}(\boldsymbol{L}^\top \! \boldsymbol{L} )$ where $\boldsymbol{A}_{\boldsymbol{L}}^\top $ and $\boldsymbol{L}$ are defined in \eqref{AL} and \textit{Theorem \ref{pole placement unknown systems}}, respectively.\\
			\eIf{$\rho_0 < 1$}{
				Return $\rho_0$ and $\mu_0$. Break.
			}{
				$r_{\mathcal{D}} = r_{\mathcal{D}}/2$. 
			}			
		}
		
	}{
		The rank condition \eqref{output rank condition} is not satisfied. Break.
	}
	%\end{algorithmic}
\end{algorithm}
\begin{remark}
	\textbf{Algorithm \ref{algorithm2}} reduces the LMI region by setting $r_{\mathcal{D}} = r_{\mathcal{D}}/2$, however this is not unique. Other ways of reducing or adjusting $r_{\mathcal{D}}$ are also feasible, such as $r_{\mathcal{D}} = r_{\mathcal{D}} - a$ where $a$ is a positive number. In addition, Algorithm \ref{algorithm2} requires that the collected data come from system \eqref{Eq:LAS}, that is, an offsets system without noise and disturbance. However, the actual data is derived from system \eqref{Eq:LDAS} with noise and disturbance. It requires Algorithm \ref{algorithm2} to be robust. One can select a smaller $\rho_0$ or expand the LMI region $\mathcal{D}$ into a robust LMI region\cite{hypiusova2019robust}. 
\end{remark}

\bibliographystyle{IEEEtran}
\bibliography{mybibfile} %IEEEabrv instead of IEEEfull

\vfill

\end{document}